\documentclass{article}
\usepackage[utf8]{inputenc}
\usepackage{amsmath,amsfonts,amssymb,mathrsfs}
\usepackage[margin=1in]{geometry}
\usepackage[noend]{algpseudocode}
\usepackage{setspace}
\usepackage{algorithm}
\usepackage{xcolor}
\usepackage{aliascnt}
\usepackage{hyperref}
\hypersetup{
    colorlinks=true,
    linkcolor=blue,
    filecolor=magenta,      
    urlcolor=cyan
}
\usepackage{tikz}
\usetikzlibrary{external,calc,arrows,decorations.pathmorphing,positioning,shapes}
\usepackage{thmtools}
\usepackage{thm-restate}
\usepackage{amsthm}

\theoremstyle{plain}
\newtheorem{theorem}{Theorem}
\newaliascnt{lemma}{theorem}
\newtheorem{lemma}[lemma]{Lemma}
\aliascntresetthe{lemma}

\newtheorem{invariant}{Invariant}

\theoremstyle{definition}
\newtheorem{definition}{Definition}
\newtheorem{observation}{Observation}
\newtheorem{remark}{Remark}

\numberwithin{equation}{lemma}

\usepackage{mathtools}

\newcommand{\rb}[2]{\raisebox{#1 mm}[0mm][0mm]{#2}}

\newcommand{\ignore}[1]{}
\newcommand{\ang}[1]{\left\langle #1 \right\rangle}
\newcommand{\bydef}{\stackrel{\mathrm{def}}{=}}

\newcommand{\ack}{\textbf{ack}}
\newcommand{\last}{\ensuremath{\operatorname{last}}}
\newcommand{\complete}{\ensuremath{\operatorname{complete}}}
\newcommand{\maxlast}{\ensuremath{\overline{\last}}}

\def\E{\ensuremath{\mathbf{E}}}
\def\BB{\ensuremath{\mathsf{BB}}}
\def\sgn{\ensuremath{\mathrm{sgn}}}

\def\CORR{\ensuremath{\operatorname{corr}}}
\def\DEV{\ensuremath{\operatorname{dev}}}
\def\PDEV{\ensuremath{{\DEV^{(p)}}}}
\def\PCORR{\ensuremath{\CORR^{(p)}}}

\newcommand{\pBB}{\ensuremath{\BB^{(p)}}}

\newcommand{\XMAX}{\ensuremath{X_{\mathrm{max}}}}
\newcommand{\KMAX}{\ensuremath{K_{\mathrm{max}}}}

\newcommand{\BrachaAgreement}{\textsf{Bracha-Agreement}}
\newcommand{\CoinFlip}{\ensuremath{\textsf{Coin-Flip}}}
\newcommand{\ReliableBroadcast}{\textsf{Reliable-Broadcast}}

\newcommand{\IteratedBlackboard}{\textsf{Iterated-Blackboard}}

\newcommand{\ceil}[1]{\left\lceil{#1}\right\rceil}

\newcommand{\poly}{\operatorname{poly}}

\title{Byzantine Agreement in Polynomial Time\\
with Near-Optimal Resilience\thanks{This work was supported by NSF grant CCF-1815316.}}

\author{Shang-En Huang\\
University of Michigan
\and
Seth Pettie\\
University of Michigan
\and 
Leqi Zhu\\
University of Michigan}
\date{}

\begin{document}
\maketitle

\begin{abstract}
It has been known since the early 1980s that Byzantine Agreement
in the full information, asynchronous model is impossible
to solve deterministically against even one crash fault~\cite{FischerLP85},
but that it \emph{can} be solved with probability 1~\cite{Ben-Or83},
even against an adversary that controls the scheduling of all messages
and \emph{corrupts} up to $f<n/3$ players~\cite{Bracha1987}.  The main downside of~\cite{Ben-Or83,Bracha1987} is that they terminate in 
$2^{\Theta(n)}$ rounds in expectation whenever $f=\Theta(n)$.

King and Saia~\cite{KingS2016,KingS2018} developed a polynomial protocol (polynomial rounds, polynomial computation) that is resilient to $f < (1.14\times 10^{-9})n$ 
Byzantine faults.  The new idea in their protocol
is to detect---and \emph{blacklist}---coalitions of likely-bad players 
by analyzing the deviations of random variables generated by those players 
over many rounds.  

In this work we design a simple 
collective coin-flipping protocol 
such that if any coalition of faulty players repeatedly does \emph{not}
follow protocol, then they will eventually be 
detected by one of two simple statistical 
tests.  
Using this coin-flipping protocol, 
we solve Byzantine Agreement in a polynomial
number of rounds, even in the presence of
up to $f<n/4$ Byzantine faults.  
This comes close to the $f<n/3$ upper bound 
on the maximum number of faults~\cite{BrachaT85,FischerLM86,LamportSP82}.
\end{abstract}

\section{Introduction}\label{sect:introduction}

The field of \emph{forensic accounting} is concerned with 
the detection of \emph{fraud} in financial transactions, 
or more generally, 
finding evidence of fraud, malfeasance, or fabrication in data sets.  
Some examples include detecting faked digital images~\cite{BonettiniBMT20},
suspicious reports of election data~\cite{Roukema14}
and political fundraising~\cite{GamermannA17},
fraudulent COVID numbers,\footnote{\url{https://theprint.in/opinion/benfords-law-detects-data-fudging-so-we-ran-it-through-indian-states-covid-numbers/673085/}}
and manipulated economic data~\cite{TildenJ12,Kauko19,Abrantes-Metz13} 
via Newcomb-Benford's law~\cite{KilaniG21}, 
detecting fabricated data sets\footnote{\url{http://datacolada.org/98}} 
in social science research~\cite{Simonsohn13,SimonsohnSN15},
or detecting match-fixing in sumo wrestling~\cite{DugganL02}.

Theoretical computer science has a strong tradition of embracing a fundamentally \emph{adversarial} view of the universe that borders on being outright paranoid.  Therefore it is somewhat surprising 
that TCS is, as a whole,  
credulous when it comes to adversarial manipulation of data and transactions.
In other words, fraud detection does not play a significant part in most 
algorithm design, \emph{even in multi-party 
models that explicitly posit the existence of malicious parties}.

\medskip

To our knowledge, the only work in TCS that has explicitly 
adopted a forensic accounting mindset is 
King and Saia's~\cite{KingS2016,KingS2018} 
breakthrough in Byzantine Agreement in the most challenging model:
the full-information (no crypto) asynchronous model
against an adaptive adversary.  In this problem there are $n$
players, each with initial input bits in $\{-1,1\}$, up to $f$
of which may fail (i.e., be \emph{adaptively 
corrupted} by the adversary) and behave arbitrarily.  
They must each \textbf{decide} on a bit 
in $\{-1,1\}$ subject to:
\begin{description}
\item[Agreement:] All non-corrupted players \textbf{decide} the same value $v$.
\item[Validity:] If all players begin with the same value $v$, all non-corrupted players \textbf{decide} $v$.
\end{description}
See \autoref{sect:themodel} for details of the model.
Prior to King and Saia's work~\cite{KingS2016,KingS2018},
it was known from Bracha~\cite{Bracha1987} (see also Ben-Or~\cite{Ben-Or83})
that the problem could be solved with probability 1 in $2^{\Theta(n)}$ 
time in expectation even if $f<n/3$ players fail, 
that $f<n/3$ cannot be improved~\cite{LamportSP82,BrachaT85,FischerLM86},
and by Fischer, Lynch, and Patterson's impossibility result~\cite{FischerLP85}, that no
\emph{deterministic} protocol exists even against a single crash failure.

King and Saia~\cite{KingS2016} reduce the problem to a certain coin-flipping game,
in which all players---good and adversarial---attempt to generate a (global) 
unbiased coin flip and agree on its outcome. Coin flipping games
have been studied extensively under adversarial manipulation (see Section~\ref{sect:related-work}),
but the emphasis is always on bounding the 
power of the adversarial players
to \emph{bias} 
the coin flip in their desired direction.  
King and Saia recognized that the  
primary \underline{long term} advantage 
of the adversary is \emph{anonymity}.
In other words, it can bias 
the outcome of 
coin flips at will, 
in the short term, 
but its advantage
simply evaporates if good players 
can merely \emph{identify} who the adversarial players are,
by detecting likely fraud via a statistical 
analysis of their transactions.
Good players can \emph{blacklist} (ignore) the adversarial players, removing their influence 
over the game.
If a sufficient number of fraudulent players are blacklisted, collective
coin-flipping by a set of good players becomes easy.

The journal version of King and Saia's work~\cite{KingS2016}
presents two methods for blacklisting players, which leads
to different fault tolerance levels.  The first protocol
has a polynomial round complexity and requires a polynomial
amount of local computation; it is claimed to be resilient
to $f<(4.25 \times 10^{-7})n$ Byzantine faults.  The second
protocol is tolerant to $f<n/400$ Byzantine faults, 
but requires exponential local computation.  
In response to some issues raised by Melynyk, Wang, and Wattenhofer (see Melynyk's Ph.D. thesis~\cite[Ch.~6]{Melnyk20}),
King and Saia~\cite{KingS2018} published a corrigendum, 
reducing the tolerance of the first protocol
to $f < (1.14\times 10^{-9})n$.

\begin{table}[]
    \centering
    \begin{tabular}{l|l|l}
    \textsf{Citation} & \textsf{Byzantine Faults ($f$)} & \textsf{Expected Rounds / Computation Per Round}\\\hline\hline
    Fischer, Lynch, Patterson \hfill 1983 & $f\geq 1$ & impossible deterministically\\\hline
    \cite{LamportSP82,BrachaT85,FischerLM86} \hfill 1982 & $f\geq n/3$ & impossible, even with randomization\\\hline
    \rb{-3}{Ben-Or} \hfill \rb{-3}{1983} & $f<n/5$ & $\exp(n)$ / $\poly(n)$\\
                    & $f<O(\sqrt{n})$ & $O(1)$ / $\poly(n)$\\\hline
    Bracha \hfill 1984 & $f<n/3$ & $\exp(n)$ / $\poly(n)$\\\hline
    \rb{-3}{King \& Saia} \hfill \rb{-3}{2016} & $f<n/400$ & $\poly(n)$ / $\exp(n)$\\
                             & $f<n/(1.14^{-1} \times 10^9)$ & $\poly(n)$ / $\poly(n)$\\\hline 
    \textbf{new} \hfill 2021 & $f<n/4$ & $\poly(n)$ / $\poly(n)$\\\hline\hline
    \end{tabular}
    \caption{Byzantine Agreement in the full information model against an adaptive adversary.}
    \label{tab:history}
\end{table}

\subsection{New Results}\label{sect:new-results}

In this paper we solve Byzantine Agreement in the full-information,
asynchronous model against an adaptive adversary, by 
adopting the same forensic accounting paradigm 
of King and Saia~\cite{KingS2016}.  We design 
a coin-flipping protocol and two simple statistical
tests such that if the Byzantine players continually
foil attempts to flip a fair coin, they will be
detected in a polynomial number of rounds 
by at least one of the tests, so long
as $f < n/4$.
(The tests measure individual deviation in 
$l_2$ norm and pair-wise correlation.)
Our analysis is tight inasmuch as these two particular
tests may not detect anything when $f\geq n/4$.

One factor contributing to the low resiliency of 
King and Saia's protocols~\cite{KingS2016,KingS2018} is that two good players 
may blacklist different sets of players, making
it easier for the adversary to induce disagreements
on the outcome of the shared coin flip.  
A technical innovation in our protocol is a method
to drastically reduce the level of disagreement between
the views of good players.  
First, we use a \emph{fractional}
blacklisting scheme.
Second, to ensure better consistency across 
good players, we extend King and Saia's~\cite{KingS2016} 
\emph{Blackboard} to 
an \emph{Iterated Blackboard} primitive that 
drastically reduces good players' disagreements of 
the historical transaction 
record by allowing retroactive corrections to the record.

\subsection{Related Work}\label{sect:related-work}

The approach of King and Saia~\cite{KingS2016} was foreshadowed 
several years earlier by Lewko~\cite{Lewko11}, who showed that 
protocols broadly similar to Ben-Or and Bracha must take an exponential 
number of rounds.  The key assumption is that messages are taken
at face value, without taking into account the \emph{identity}
of the sender, nor the \emph{history} of the sender's messages.

Byzantine agreement has been studied in synchronous and asynchronous
models, against computationally bounded or unbounded adversaries,
and with adaptive or non-adaptive adversaries. (In particular, a special case of the problem that restricts attention to crash failures, called \emph{consensus}, has been very extensively studied.) 
We refer the 
reader to~\cite{aspnes2003randomized,attiya2008tight,bjbo1998,ben2006byzantine,correia2011byzantine,king2011breaking} for some key 
results and surveys of the literature.
A result that is fairly close to ours is that of 
Kapron et al.~\cite{KapronKKSS10}.
They proved that against a \emph{non-adaptive} adversary (all corruptions made in advance)
Byzantine agreement can be solved asynchronously, against $f<n/(3+\epsilon)$ faults.

\medskip

Collective coin flipping has an illustrious history in computer science,
as it is a key concept in cryptography, distributed computing,
and analysis of boolean functions.  The problem was apparently first raised
by Blum~\cite{Blum81}, who asked how two mutually untrusted parties could
flip a shared coin over the telephone.  His solution used cryptography.
See~\cite{Cleve86,HaitnerT17,MoranNS16,BeimelOO15,Dachman-SoledMM14,HaitnerMO18,BuchbinderHLT21} for some recent work on coin flipping using cryptography.

\medskip

Ben-Or and Linial~\cite{Ben-OrL85} initiated a study of \emph{full information}
protocols for coin-flipping.  The players broadcast messages one-by-one in 
a specific order, and the final coin flip is a function of these messages.
The goal is to minimze the \emph{influence} of a coalition of $k$ bad players,
which is, roughly speaking, the amount by which they can bias the outcome towards
\emph{heads} or \emph{tails}.  Ben-Or and Linial's~\cite{Ben-OrL85} protocol
limits $k<n^{\log_3 2}$ bad players to influence $O(k/n)$.  Saks~\cite{Saks89}
and Ajtai and Linial~\cite{AjtaiL93} improved it to $O(k/n)$ influence with up to 
$k=O(n/\log n)$ players, 
and Alon and Naor~\cite{AlonN93} achieved optimum $O(k/n)$ influence for $k$ even linear in $n$.  The message size in these protocols is typically more than a single bit. 
If only single-bit messages are allowed and each player speaks once, 
the problem is equivalent to bounding the influence of variables in a boolean function~\cite{KahnKL88}.
Russel, Saks, and Zuckerman~\cite{RussellSZ02} considered parallel coin-flipping
protocols. The proved that any protocol that uses 1-bit messages and is resilient 
to linear-size coalitions must use $\Omega(\log^* n)$ rounds.

Aspnes~\cite{Aspnes98} considered a sequential coin-flipping game where $n$ coins
are flipped sequentially and the outcomes broadcast, but up to $t$ of these may be
\emph{suppressed} by the adversary.  Regardless of which function is used to map
the coin-flip sequence to a shared coin, the adversary can bias it whenever $t=\Omega(\sqrt{n})$.  
Very recently Haitner and Karidi-Heller~\cite{HaitnerK20}
resolved the complexity of  Ben-Or-Linial-type sequential coin flipping games
against an \emph{adaptive} adversary, that can corrupt players at will, as information is revealed. 
They proved that \emph{any}
such shared coin can be fixed to a desired outcome with probability $1-o(1)$  
by adaptively corrupting $\tilde{O}(\sqrt{n})$ parties.

\subsection{Organization}

In Section~\ref{sect:preliminaries} we review the model,
the reliable broadcast primitive, and Bracha's Byzantine agreement protocol,
and introduce the \emph{Iterated Blackboard} primitive, which generalizes~\cite{KingS2016,Kimmett2020}.

In Section~\ref{sect:iterated-coin-flipping} we begin with a simplified iterated 
coin-flipping
game and then proceed to study a more complicated iterated coin-flipping game that
can be implemented in the asynchronous distributed model and used within Bracha's algorithm.

Appendix~\ref{sect:appendix-prelim-proofs} contains proofs from Section~\ref{sect:preliminaries}
on reliable broadcast and the iterated blackboard.
Appendix~\ref{sect:tailbounds} reviews some standard concentration
inequalities and other theorems.
Appendix~\ref{section:rising-tide-proofs} contains some 
proofs showing that a certain fractional matching algorithm
has a Lipschitz property.

\section{Preliminaries}\label{sect:preliminaries}

\subsection{The Model}\label{sect:themodel}

There are $n$ \emph{processes}, $p_{1},\dots,p_{n}$, and $2n^2$ 
\emph{message buffers}, 
$\operatorname{In}_{j\rightarrow i}$ and 
$\operatorname{Out}_{i\rightarrow j}$ for all $i,j\in[n]$.
All processes are initially \emph{good} (they obey the protocol)
and the adversary may dynamically
\emph{corrupt} up to $f$ processes.
A \emph{bad}/\emph{corrupted} process is under complete 
control of the adversary and may behave arbitrarily.
The adversary controls the pace at which progress is made by scheduling
two types of \emph{events}.
\begin{itemize}
    \item A $\operatorname{compute}(i)$ event lets $p_i$ process all
    messages in the buffers $\operatorname{In}_{j\rightarrow i}$, 
    deposit new messages in $\operatorname{Out}_{i\rightarrow j}$, 
    and change state.
    \item A $\operatorname{deliver}(i,j)$ event removes a message from
    $\operatorname{Out}_{i\rightarrow j}$ and moves it to
    $\operatorname{In}_{i\rightarrow j}$.
\end{itemize}
Note that the adversary may choose a malicious order of events, but
cannot, for example, misdeliver or forge messages.
The adversary must eventually allow some good process to make progress.
In particular, we can assume without loss of generality that
the scheduling sequence is of the form
\[
A_0,A_1,A_2,\ldots,
\]
where each $A_k$ contains a finite number of events, including
either the first $\operatorname{compute}(i)$ event for good process $i$,
or the delivery of a message from $\operatorname{Out}_{i\rightarrow j}$
to $\operatorname{In}_{i\rightarrow j}$ 
followed later by $\operatorname{compute}(j)$, for some good process $i$.
Each $A_k$ can contain an arbitrary number of $\operatorname{compute}$ events
for bad processes, or the delivery of messages sent by bad processes.

The adversary is computationally unbounded and is aware, at all times, of 
the internal state of all processes.  Thus, cryptography is not helpful, but
randomness potentially is, since the adversary cannot predict the outcome 
of future coin flips.

In this model, the \emph{communication time} or \emph{latency} is defined w.r.t.~a hypothetical execution in which all local computation occurs instantaneously and all messages have latency in $[0,\Delta]$.  The latency of the algorithm is $L$ if all non-corrupt processes finish by time $L\Delta$.  Note that in this hypothetical, $\Delta$ is unknown and cannot influence the execution of the algorithm.

\subsubsection{Reliable Broadcast}\label{sect:reliable-broadcast}

The goal of \ReliableBroadcast{} is to simulate a broadcast channel
using the underlying point-to-point message passing system.
In Byzantine Agreement protocols, each process initiates a series of
$\ReliableBroadcast$s.  Call $m_{p,\ell}$ the $\ell$th message 
broadcast by process $p$.  

\begin{restatable}{theorem}{reliablebroadcastthm}\label{thm:reliable-broadcast}
If a good process $p$ initiates the \ReliableBroadcast{} of $m_{p,\ell}$, then
all good processes $q$ eventually \emph{accept} $m_{p,\ell}$.
Now suppose a bad process $p$ does so and some good $q$ \emph{accepts} $m_{p,\ell}$.
Then all other good $q'$ will eventually accept $m_{p,\ell}$, and
no good $q'$ will accept any other $m'_{p,\ell}\neq m_{p,\ell}$.
Moreover, all good processes accept $m_{p,\ell-1}$ before $m_{p,\ell}$, if $\ell>1$.
\end{restatable}

The property that $m_{p,\ell}$ is only accepted after $m_{p,\ell-1}$ is accepted is sometimes called FIFO broadcast.  This property is explicitly used in the \IteratedBlackboard{} algorithm
outlined in \autoref{sect:blackboard}.
See Appendix~\ref{sect:proof-reliable-broadcast} 
for a proof of \autoref{thm:reliable-broadcast}.

\begin{algorithm}
    \caption{$\ReliableBroadcast(p,\ell)$}
    \label{alg:reliable-broadcast}
    \begin{algorithmic}[1]
        \State{\textbf{if} $\ell>1$ \textbf{then} \textbf{wait} until $m_{p,\ell-1}$ has been accepted.}
        
        \State{\textbf{if} I am process $p$ \textbf{then} generate $m_{p,\ell}$ and send 
        $(\operatorname{init},m_{p,\ell})$ to all processes.}
        \State{\textbf{wait} until receipt of one $(\operatorname{init},m_{p,\ell})$ 
        message from $p$, or more than $(n+f)/2$ $(\operatorname{echo},m_{p,\ell})$
        messages, or $f+1$ $(\operatorname{ready},m_{p,\ell})$ messages.
        {\par}
        send $(\operatorname{echo}, m_{p,\ell})$ to all processes.}
        \State \textbf{wait} until the receipt of $(n+f)/2$ $(\operatorname{echo},m_{p,\ell})$
        messages or $f+1$ $(\operatorname{ready},m_{p,\ell})$ messages.
        {\par}
        send $(\operatorname{ready},m_{p,\ell})$ to all processes.
        \State \textbf{wait} until receipt of $2f+1$ $(\operatorname{ready},m_{p,\ell})$ messages. 
        {\par}
        \textbf{accept} $m_{p,\ell}$.
    \end{algorithmic}
\end{algorithm}

\subsubsection{Validation and Bracha's Protocol}

Consider a protocol $\Pi$ of the following form.
In each round $r$ , each process reliably 
broadcasts its \emph{state} 
to all processes, waits until it has 
accepted at least $n-f$ 
\emph{validated} messages from round $r$, 
then processes all validated messages, 
changes its state, and advances to round $r+1$.
A good process \emph{validates} 
a round-$r$ state (message) $s_{q,r}$ accepted 
from another process $q$ 
only if 
(i) it has validated the state $s_{q,r-1}$ of $q$ at round $r-1$,
and
(ii) it has accepted $n-f$ messages that, if they \emph{were}
received by a correct $q$, 
would cause it to transition from $s_{q,r-1}$ to $s_{q,r}$.
The key property of validation (introduced by \cite{Bracha1987}) is:

\begin{lemma}
A good process $p$ validates the message of another process $q$
in an admissible execution $\alpha$ of $\Pi$ if and only if there is an execution
$\beta$ of $\Pi$ in which $q$ is a good process and the state of
every other good process (including $p$) is the same in $\alpha$
and $\beta$ (with respect to their validated messages).
\end{lemma}

To recap, reliable broadcast prevents the adversary from sending
conflicting messages to different parties (i.e., it is forced to 
participate as if the communication medium were a broadcast channel)
and the validation mechanism forces its internal state transitions
to be consistent with the protocol.  Its remaining power is limited
to (i) substituting deterministic outcomes for coin flips in bad processes,
(ii) dynamic corruption of good processes, and (iii) malicious scheduling.

Bracha's protocol improves the resilience of 
Ben-Or's protocol to the optimum $f<n/3$.
Each process $p$ initially holds a value $v_p\in \{-1,1\}$. 
It repeats the same steps until it \textbf{decides} a value $v\in\{-1,1\}$ (Line~\ref{line:decide-v}).
As we will see, if some process \textbf{decides} $v$, 
all good processes
will \textbf{decide} $v$ in this or the following iteration. 
Thus, good processes
continue to participate in the protocol 
until all other good processes have
executed Line~\ref{line:decide-v}.
Here $\sgn(x) = 1$ if $x\geq 0$ and $-1$ if $x<0$.

\begin{algorithm}[H]
    \caption{\BrachaAgreement() \ \emph{from the perspective of process $p$}}
\begin{algorithmic}[1]
\Require{$v_p\in \{-1,1\}$.}
\Loop
\State
reliably broadcast $v_{p}$ and \textbf{wait} until $n-f$ messages are validated from some processes $S$.
{\par \enspace}
set $v_p := \sgn(\sum_{q\in S} v_q)$. \label{line:Bracha1}

\State
reliably broadcast $v_{p}$ and \textbf{wait} until $n-f$ messages are validated.  
{\par \enspace}
\textbf{if} more than $n/2$ messages have some value $v$ \textbf{then} set $v_p := (\operatorname{dec}, v)$.
\label{line:Bracha2}

\State
reliably broadcast $v_{p}$ and \textbf{wait} until $n-f$ messages 
are validated. 
{\par \enspace}
let $x_p$ be the number of $(\operatorname{dec},v)$ messages validated by $p$.
\label{line:Bracha3}

\If{$x_p \geq 1$}\label{line:firstif}
    \State set $v_p := v$. \label{line:fplus1}
\EndIf

\If{$x_p \geq f+1$}\label{line:secondif}
    \State \textbf{decide} $v$. \label{line:decide-v}
\EndIf

\If{$x_p=0$} \label{line:thirdif}
    \State $v_p := \CoinFlip()$. \label{line:coinflip}
\Comment{Returns value in $\{-1,1\}$.}
\EndIf
\EndLoop
\end{algorithmic}
\end{algorithm}

\paragraph{Correctness.}
Suppose that at the beginning of an iteration, 
there is a set of at least $(n+f+1)/2$ good processes 
who agree on a value $v\in\{-1,1\}$.\footnote{Note that this is always numerically possible 
since $(n+f+1)/2 \leq n-f$ with equality if $f=(n-1)/3$.}
It follows that in Line~\ref{line:Bracha1}, every process
hears from at least $(n+f+1)/2-f > (n-f)/2$ of these good processes,
i.e., a strict majority in any set of $n-f$.
Thus, every good process broadcasts $v$ in Line~\ref{line:Bracha2},
and due to the validation mechanism, any bad process that
wishes to participate in Line~\ref{line:Bracha2} \emph{also}
must broadcast $v$.  Thus, every good process $p$ will eventually
validate $n-f > n/2$ votes for $v$ and set $v_p := (\operatorname{dec},v)$
indicating it is prepared 
to decide $v$ in this iteration. 
By the same reasoning,
every good process $p$ will set 
$x_p := n-f \geq f+1$ and 
\textbf{decide} $v$ in Line~\ref{line:decide-v}.

It is impossible for $p$ to validate two 
messages $(\operatorname{dec},v)$ and $(\operatorname{dec},v')$
in Line~\ref{line:Bracha3} with $v\neq v'$.  To validate
such messages, $p$ would need to receive strictly greater than
$n/2$ ``$v$'' and ``$v'$'' messages in Line~\ref{line:Bracha2},
meaning some process  successfully broadcast two distinct 
messages with
the same timestamp.  By \autoref{thm:reliable-broadcast} this is impossible.

Now suppose that in some iteration $p$ \textbf{decides} $v$ in Line~\ref{line:decide-v}.
This means that $p$ validated $n-f$ messages 
in Line~\ref{line:Bracha3} and set $x_p \geq f+1$.  
Every other good process $q$ must have validated
at least $n-2f$ of the messages that $p$ validated, 
and therefore set $x_q \geq 1$, forcing it to 
set $v_q := v$ in Line~\ref{line:fplus1}.  Thus, at
the beginning of the next iteration $n-f$ good processes 
agree on the value $v$ and all \textbf{decide} 
$v$ (Line~\ref{line:decide-v}) in that iteration.\footnote{Bracha~\cite{Bracha1987} sets the thresholds in Line~\ref{line:firstif} and \ref{line:secondif} to be $f+1$ and $2f+1$.
The idea was to guarantee that if $x_p\geq f+1$ then at least one \emph{good} process sent $p$ a $(\operatorname{dec},v)$ message. However, 
because of the validation mechanism this is not important.  A corrupt process can \emph{try} to send a $(\operatorname{dec},v)$ message but it will not be validated unless $v$ does, in fact, have a strict majority ($>n/2$) of messages sent in Line~\ref{line:Bracha2}.}

The preceding paragraphs establish correctness.
Turning to efficiency, consider any iteration in which
no process \textbf{decides} $v$ in Line~\ref{line:decide-v}.
We can partition the good population into $G_{\mbox{\scriptsize\ref{line:fplus1}}}$ and $G_{\mbox{\scriptsize\ref{line:coinflip}}}$,
depending on whether they execute 
Line~\ref{line:fplus1} (setting $v_p := v$) 
or Line~\ref{line:coinflip}.
If a sufficiently large number of calls to $\CoinFlip()$ 
made by $G_{\mbox{\scriptsize\ref{line:coinflip}}}$-processes returns $v$
(specifically, $(n+f+1)/2 - |G_{\mbox{\scriptsize\ref{line:fplus1}}}|$)
then by the argument above, all processes will \textbf{decide} $v$ 
(Line~\ref{line:decide-v}) in the next iteration.
Call this happy event $\mathcal{E}$.
If $G_{\mbox{\scriptsize\ref{line:fplus1}}}=\emptyset$ 
then both values of $v$ are acceptable,
which just increases the likelihood of $\mathcal{E}$.

Bracha~\cite{Bracha1987} and Ben-Or~\cite{Ben-Or83} 
implement $\CoinFlip$ by each process
privately flipping an independent, unbiased coin.
Thus, for any $f<n/3$, $\Pr(\mathcal{E}) \geq 2^{-(n-f-1)}$
and the expected number of iterations is at most $2^{\Theta(n)}$.
If there were a mechanism to implement $\CoinFlip$ as 
a roughly unbiased \emph{shared} coin
(all processes in $G_{\mbox{\scriptsize\ref{line:coinflip}}}$ see the same
value; see Rabin~\cite{Rabin83} and Toueg~\cite{Toueg84}), then $\Pr(\mathcal{E})$ is constant and we only need $O(1)$ iterations in expectation.
Efficient collective coin-flipping is 
therefore the heart of the 
Byzantine Agreement problem in this model.

\subsection{The Iterated Blackboard Model}\label{sect:blackboard}

King and Saia~\cite{KingS2016} implemented a $\CoinFlip()$ routine
using a \emph{blackboard} primitive, which weakens the power of the
scheduling adversary to give drastically 
different views to different processes.\footnote{For example, in Line~\ref{line:Bracha1} of \BrachaAgreement, the scheduling adversary
can show $p$ \emph{any} $n-f$ messages $S$, and therefore have significant
control over the value of $\sgn(\sum_{q\in S}v_q)$.}
Their blackboard protocol is resilient to $f<n/4$ faults.  
Kimmett~\cite{Kimmett2020} simplified and improved this protocol 
to tolerate $f<n/3$ faults.  
In this section, we describe a useful extension of the 
Kimmett-King-Saia style 
blackboard that \emph{further reduces} the 
kinds of disagreements that good processes can have.

In the original model~\cite{KingS2016,Kimmett2020}, a \emph{blackboard} is an $m \times n$ matrix $\BB$, initially all blank ($\perp$), 
such that column $\BB(\cdot,i)$ is only written to by process $i$.
Via reliable broadcasts, process $i$ attempts to sequentially write non-$\perp$ values to $\BB(r,i)$, $r\in [m]$.  The scheduling
power of the adversary allows it to control the rate at which 
different processes write values.  Because there could be up to $f$ crash-faults, no process can count on $\BB$ containing more than $n-f$
\emph{complete} columns (those $i$ for which $\BB(m,i)\neq \perp$).
The final $\BB$-matrix may therefore contain up to $f$ \emph{partial} columns.

The main guarantee of~\cite{KingS2016,Kimmett2020} is that every 
process $p$ has a mostly accurate \emph{view} $\BB^{(p)}$ 
that agrees with 
the ``true'' blackboard $\BB$ in all but at most $f$ locations.
In particular, the last non-$\perp$ entry of each partial column
in $\BB$ may still be $\perp$ in $\BB^{(p)}$.
If we were to generate a sequence of blackboards with \cite{KingS2016,Kimmett2020}, the views from two processes
could differ by $f$ locations in \emph{each} blackboard.

\medskip

An \emph{iterated blackboard} is an endless series 
$\BB = (\BB_1,\BB_2,\ldots)$ of $m \times n$ blackboards,
such that process $i$ only attempts to write its column 
in $\BB_t$ once it completes participation in $\BB_{t-1}$. 
After $p$ regards $\BB_{t}$ as complete, $p$ obtains a view
of the full history
$\BB^{(p,t)} = (\BB_1^{(p,t)},\ldots, \BB_t^{(p,t)})$
that differs from $(\BB_1,\ldots,\BB_t)$ in $f$ locations 
\emph{in total}.  As a consequence, $\BB^{(p,t-1)}$ may
not be identical to the first $t-1$ matrices of $\BB^{(p,t)}$, 
i.e., $p$ could record ``retroactive'' updates to previous
matrices while it is actively participating in the construction 
of $\BB_t$.

The following theorem is proved in Appendix~\ref{sect:appendix-blackboard}.

\begin{restatable}{theorem}{blackboardthm}
\label{thm:blackboard}
There is a protocol for $n$ processes to generate an iterated 
blackboard $\BB$ that is resilient to $f<n/3$ Byzantine failures.
For $t \geq 1$, the following properties hold:
\begin{enumerate} 
    \item \label{thm:blackboard-full-partial} 
    Upon completion of the matrix $\BB_t$, each column consists of a 
    prefix of non-$\bot$ values and a suffix of all-$\bot$ values. 
    Let $\last(i) = (t',r)$ be the position of the last value written by
    process $i$, i.e., $\BB_{t'}(r,i)\neq \bot$ and if $t'<t$ then $i$ has not 
    written to any cells of $\BB_t$.
    When $\BB_{t}$ is complete, 
    it has at least $n-f$ full columns and up to
    $f$ partial columns.  
    \item \label{thm:blackboard-ambiguous} 
    Once $\BB_{t}$ is complete, each process $p$ forms a history 
    $\BB^{(p,t)} = (\BB_1^{(p,t)},\ldots,\BB_t^{(p,t)})$ 
    such that for every $t'\in [t]$, $i\in [n]$, $r\in [m]$,
    \[
    \BB_{t'}^{(p,t)}(r,i) \left\{
    \begin{array}{l@{\hspace{1cm}}l}
    = \BB_{t'}(r,i)       & \mbox{if $\last(i)\neq (t',r)$}\\
    \in \{\BB_{t'}(r,i), \bot\} & \mbox{otherwise}
    \end{array}\right.
    \]

    \item If $q$ writes any non-$\bot$ value to $\BB_{t+1}$,
    then by the time any process $p$ fixes $\BB^{(p,t+1)}$, 
    $p$ will be aware of $q$'s view $\BB^{(q,t)}$ of the history up to blackboard $t$.
\end{enumerate}
\end{restatable}

\section{Iterated Coin Flipping Games}\label{sect:iterated-coin-flipping}

We begin in \autoref{sect:simplified-coin-flipping}
with a simplified coin-flipping game 
and extend it in \autoref{sect:real-coin-flipping-game} 
to the real coin-flipping game we use to implement 
$\CoinFlip()$ in $\BrachaAgreement$.
In the real coin-flipping 
game we assign \emph{weights} to the processes,
which is a measure of trustworthiness.  
\autoref{sect:weight-update}
explains how the weights are updated 
and \autoref{sect:bounding-error}
bounds numerical inconsistencies in different processors views.

\subsection{A Simplified Game}\label{sect:simplified-coin-flipping}

In this game there are $n$ players partitioned into $n-f$ \emph{good} 
players $G$ and $f=n/(3+\epsilon)$ \emph{bad} players $B$, 
for some small $\epsilon>0$.
The good players are unaware of the partition $(G,B)$.
The game is played up to $T$ times in succession according to the 
following rules.  Let $t\in[T]$ be the current iteration.
\begin{itemize}
    \item The adversary privately picks 
            an \emph{adversarial direction} $\sigma(t)\in\{-1,1\}$.\footnote{In the context of $\BrachaAgreement$, $\sigma$ would be $-v$, where 
            $v$ is the value set by processes executing Line~\ref{line:fplus1}.}
    \item Each good player $i\in G$ picks $X_i(t)\in\{-1,1\}$ uniformly at random.  The bad players see these values then generate their
    values $\{X_i(t)\}_{i\in B}$, each in $\{-1,1\}$, as they like.
    \item If the \emph{outcome} of the coin flip,
    $\sgn(\sum_{i\in [n]} X_i(t))$,  is equal to $\sigma(t)$, the game continues to iteration $t+1$.
\end{itemize}
From the good players' perspective, the nominal goal of this game is
to eventually achieve the outcome 
$\sgn(\sum_{i\in [n]} X_i(t)) \neq \sigma(t)$, but the adversary
can easily foil this goal if $T=\poly(n)$. We consider a secondary goal:
namely to \emph{identify} bad players based solely on the historical data $\{X_i(t)\}_{i,t}$.
This turns out to be a tricky problem, but we can identify a \emph{pair} 
of processes, at least one of which is bad, w.h.p.

\begin{lemma}
\label{lem:correlation-easy}
Suppose the game does not end after $T$ iterations. 
If $T=\tilde{\Theta}((n/\epsilon)^2)$, then 
the pair $(i,j)\in [n]^2$, $i\neq j$, maximizing
\[
\ang{X_{i},X_{j}} = \sum_{t=1}^T X_{i}(t)X_{j}(t)
\]
has $B\cap \{i,j\}\neq\emptyset$.
\end{lemma}

\begin{proof}
If $i,j\in G$ are good, by a Chernoff-Hoeffding bound 
(\autoref{thm:hoeffding}, Appendix~\ref{sect:tailbounds})
$\ang{X_i,X_j} \leq \beta = \tilde{O}(\sqrt{T})$ with high probability, 
thus every pair whose inner product exceeds $\beta$ must contain
at least one bad process.
We now argue that there exists an $i^\star,j^\star\in B$ 
such that $\ang{X_{i^\star},X_{j^\star}}$ exceeds $\beta$.
Observe that
\begin{align}
    \sum_{(i\neq j)\in B^2} X_i(t)X_j(t) &= \left(\sum_{i \in B} X_i(t)\right)^2 - \sum_{i \in B} (X_i(t))^2 = \left(\sum_{i \in B} X_i(t)\right)^2 - f.\label{eqn:bad-corr}
\end{align}
Let $S(t)=\sum_{i\in G} X_i(t)$ be the sum of the good processes in iteration $t$.
The bad players force the sign of the sum to be $\sigma(t)$,
i.e., $\sgn(S(t)+\sum_{i \in B} X_i(t))\sigma(t) = 1$. Thus,
\begin{align}
\left(\sum_{i \in B} X_i(t)\right)^2 
    &\geq 
    \begin{cases}
    (S(t))^2            &\textrm{if $\sgn(S(t))\neq \sigma(t)$}\\
    0                   &\textrm{otherwise}
    \end{cases} \nonumber\\
    &= (\max \{0, -\sigma(t)S(t)\})^2.\label{eqn:bad-sum}
\end{align}
Let $Z(t) = (\max \{0, -\sigma(t)S(t)\})^2$. 
By a Chernoff-Hoeffding bound (\autoref{thm:hoeffding}, Appendix~\ref{sect:tailbounds}), 
w.h.p.~$Z(t) \leq \gamma = \tilde{O}(n)$ 
for every $t$. Moreover, since the distribution
of $S(t)$ is symmetric around the origin,
\begin{align}
\E[Z(t)] 
&\geq \frac{1}{2}\E[S(t)^2 \mid -\sigma(t)S(t)\geq 0]
= \frac{1}{2}(n-f).
\intertext{Thus, by linearity of expectation and Chernoff-Hoeffding (\autoref{thm:hoeffding}, Appendix~\ref{sect:tailbounds}), 
we have, w.h.p.,}
\sum_{t=1}^T Z(t) 
&\geq
\frac{1}{2}T(n-f) - 
\gamma \cdot \tilde{O}(\sqrt{T}) 
\;=\;
\frac{1}{2}T(n-f) - \tilde{O}(n\sqrt{T}).\label{eqn:Z-sum}
\end{align}
Combining Eqns.~(\ref{eqn:bad-corr}), (\ref{eqn:bad-sum}), and (\ref{eqn:Z-sum}), 
we have, w.h.p.,
\begin{align}
\sum_{(i\neq j)\in B^2} \ang{X_i,X_j} 
 = \sum_{t\in [T]}\left(\left(\sum_{i \in B} X_i(t)\right)^2 - f\right)
 &\geq \sum_{t\in [T]} (Z(t) - f)\nonumber\\
 &\geq \frac{1}{2}T(n-3f) - \tilde{O}(n\sqrt{T})
 \;=\; \epsilon n T/2 - \tilde{O}(n\sqrt{T}).\label{eqn:total-corr}
\end{align}
We lower bound the average correlation score
within $B$ by dividing
Eqn.~(\ref{eqn:total-corr}) 
by the $f(f-1)$ distinct pairs $i,j\in B^2$.
Using the fact that $f=n/(3+\epsilon)$, we have
\begin{align*}
    \max_{i^*,j^*\in B, i^*\neq j^*} \ang{X_{i^*},X_{j^*}} \geq \frac{1}{f(f-1)}\left(\epsilon nT/2 - \tilde{O}(n\sqrt{T})\right)
    \geq \frac{(3+\epsilon)\epsilon}{2(f-1)} T - \tilde{O}(\sqrt{T}/n)
\end{align*}
Note that the $\tilde{O}(\sqrt{T}/n)$ term 
is negligible and that
$\frac{(3+\epsilon)\epsilon}{2(f-1)} T 
\gg \beta = \tilde{O}(\sqrt{T})$
whenever $T = \tilde{\Theta}((n/\epsilon)^2)$.
\end{proof}

\subsection{The Real Coin-Flipping Game}\label{sect:real-coin-flipping-game}

In this section we describe a protocol for calling
$\CoinFlip()$ iteratively in the context of 
$\BrachaAgreement$.  
It is based on a coin-flipping
game that differs from the simplified game of 
\autoref{sect:simplified-coin-flipping} in several
respects, most of which stem from the power of the
adversarial scheduler to give good players slightly different
views of reality. The differences are as follows.
\begin{itemize}
\item  
The bad players are \emph{not} 
fixed in advance, but may be 
corrupted at various times. 
\item Rather than picking $X_i(t)\in\{-1,1\}$, the 
processes generate an iterated blackboard $\BB$ where each write is a value in $\{-1,1\}$, chosen
uniformly at random \emph{if the writing process is good}. Each blackboard $\BB_t$ has $n$ columns and $m = \Theta(n/\epsilon^2)$ rows.
When $\BB_t$ is complete, let
$X_i(t)$ be the sum of all non-$\bot$ 
values in column $\BB_t(\cdot,i)$.
Every player's view of reality is 
slightly different.
$X_i^{(p)}(t)$ refers to $p$'s most up-to-date view of $X_i(t)$, which is initially 
the sum of column $\BB^{(p,t)}_t(\cdot,i)$.
By Theorem~\ref{thm:blackboard}, 
$\sum_{i\in [n]}\left|X_i^{(p)}(t)-X_i(t)\right|\leq f$ for any $p,t$.
\item 
Each process $i$ has a \emph{weight} $w_i\in [0, 1]$, initially 1, which is non-increasing over time. At all times, the processes maintain \emph{complete agreement} on the weights of the actively participating processes, i.e., those who broadcast coin flips. 
This is accomplished as follows. 
By Theorem~\ref{thm:blackboard}(3), if any process $q$ writes to $\BB_t$,
every other process $p$ learns $\BB^{(q,t-1)}$ by the time they finish computing $\BB_t$.  Based on the history $\BB^{(q,t-1)}$, $p$ can locally compute the weight vector $(w_i^{(q)})_{i\in [n]}$ of $q$. However, due to different views of the history, $(w_i^{(q)})_{i\in [n]}$ may be slightly different than $(w_i^{(p)})_{i\in[n]}$.  
We reconcile this by defining the weight of 
each participating process based on its \emph{own} view of history, i.e.~
\[
w_i = \left\{\begin{array}{ll}
w_i^{(i)}  & \mbox{ if $w_i^{(i)} > w_{\min}$}\\
0           & \mbox{ otherwise.}
\end{array}\right.
\]
In other words, $w_i$ is drawn from the weight vector computed by process $i$. Thus, by Theorem~\ref{thm:blackboard}(3), the weight $w_i$ 
of any process participating in $\BB_t$ is \emph{common knowledge}.  
(It is fine that the weights of non-participating processes remain uncertain.)
For technical reasons, a weight is rounded down to 0 if it is less than a small threshold, $w_{\min}=\sqrt{n\ln n}/T$,
where $T$ is defined below. 
\item 
In iteration $t$, process $p$ sets its own output of $\CoinFlip()$ to be 
$\sgn\left(\sum_{i\in [n]} w_iX_i^{(p)}(t)\right)$.
If this quantity is $-\sigma(t)$ for every good process $p$,
the game ends ``naturally.''  (In the next iteration of Bracha's algorithm, all processes will {\bf decide} on a common value.)
\item 
The iterations are partitioned into $O(f)$ \emph{epochs}, 
each with $T = \Theta(n^2\ln^3 n/\epsilon^2)$ iterations, 
where the goal of each epoch is to 
either end the game naturally
or gather enough statistical evidence to 
reduce the weight of 
some processes before the next epoch begins.
This can be seen as \emph{fractional blacklisting}.

\item Because the scheduling adversary can avoid delivering 
messages from $f$ good processes, 
the resiliency of the protocol drops 
to $f = n/(4+\epsilon)$.  Any positive $\epsilon>0$ suffices,
so we can tolerate $f$ as high as $(n-1)/4$.  In some places
we simplify calculations by assuming $\epsilon \leq 1/2$.
\end{itemize}

Throughout $c$ is an arbitrarily large constant.
All ``with high probability'' bounds hold with probability $1-n^{-\Omega(c)}$. Since each process flips at most $m$ coins in each iteration,
by a Chernoff-Hoeffding bound (\autoref{thm:hoeffding}, Appendix~\ref{sect:tailbounds}) we have
\[
|X_i(t)|\leq \sqrt{cm\ln n} \bydef X_{\max}
\]
holds for all $i,t$, with high probability.
To simplify some arguments we will actually 
enforce this bound deterministically.
If $X_i(t)$ is not in the interval $[-\XMAX,\XMAX]$,
map it to the nearest value of $\pm\XMAX$.

With high probability, 
the weight updates always respect \autoref{inv:weights},
which says that the total weight-reduction of
good players is at most the weight reduction of bad players, up 
to an additive error of $\epsilon^2f/8$.
This error term arises from the fact that we are integrating 
slightly inconsistent weight vectors $(w_i^{(p)})$ for each $p$ to yield $(w_i)$.
With the assumption $\epsilon \le 1/2$, 
\autoref{inv:weights} implies that the total 
weight of good processes is always $\Omega(n)$.

\begin{invariant}\label{inv:weights}
Let $G$ and $B$ denote the set of good and bad processes at any given time. 
Then, 
\[
\displaystyle\sum_{i\in G} (1-w_i) \le \sum_{i\in B} (1-w_i) + \epsilon^2f/8.
\]
\end{invariant}

Whereas pairwise correlations alone suffice to detect bad players in the simplified game, the bad players
can win the real coin-flipping game without being detected by this particular test.  
As we will see, this can only be accomplished if 
$\{X_i(t)\}_{t\in [T]}$ differs significantly 
from a binomial distribution, for some $i\in B$.
Thus, in the real game we
measure individual deviations in the $l_2$-norm in addition to pairwise correlations.
Define $\DEV(i)$ and $\CORR(i,j)$ at the end of a particular 
epoch as below.  The iterations of the epoch are indexed by 
$t\in [T]$ and throughout the epoch the weights $\{w_i\}$ 
are unchanging.
\begin{align*}
    \DEV(i) &= \sum_{t\in [T]} (w_i X_i(t))^2,\\
    \CORR(i,j) &= \sum_{t\in [T]} w_i w_j X_i(t) X_j(t).
\end{align*}
Naturally each process $p$ estimates these quantities using
its view of the historical record; let them be $\PDEV(i)$ and $\PCORR(i,j)$.

The Gap Lemma says that if we set the deviation 
and correlation thresholds $(\alpha_T, \beta_T)$ 
properly, no good player will exceed its deviation
budget, no pairs of good players will exceed their
correlation budget, but some bad player or pair
involving a bad player will be detected by one of these tests.  One subtle point to keep in mind in this section
is that random variables that depend on the coins flipped by good players can still
be heavily manipulated by the scheduling power of the adversary.\footnote{For example, by Doob's optional stopping theorem for martingales,
it is true that $\E[X_i(t)]=0$, 
but not true that the distribution of $X_i(t)$
is symmetric around 0, or that 
it is close to binomial, 
or that we can say anything about $X_i(t)$ after 
conditioning on some 
natural event, 
e.g., that it was derived
from summing the values in a \emph{full} column of $\BB_t(\cdot,i)$.}
See~\cite[Ch.~6]{Melnyk20} for further discussion 
of this issue.


\begin{lemma}[The Gap Lemma]\label{lem:gap-lemma}
Consider any epoch in which the game does not end,
and let $\{w_i\}_{i \in [n]}$ be process weights.
Let $G$ and $B$ be the good and bad processes at the end
of the epoch. 
With high probability, 
\begin{enumerate}
    \item 
    Every good $i\in G$ has $\DEV(i) \le w_i^2 \alpha_T$, where $\alpha_T = m(T + \sqrt{T(c\ln n)^3})$. \label{lem:gap-lemma-good-dev}
    \item 
    \label{lem:gap-lemma-good-corr}
    Every pair $i,j\in G$ has $\CORR(i, j) \le w_iw_j \beta_T$, where $\beta_T = m\sqrt{T(c\ln n)^3}$.
    \item If the weights satisfy Invariant~\ref{inv:weights} and no processes were added to $B$ in this epoch, then 

    \[
    \sum_{i\in B} \max\{0,  \DEV(i) - w_i^2\alpha_T\} + \sum_{(i\neq j)\in B^2} \max\{0, \CORR(i, j) - w_iw_j\beta_T\} \ge \frac{\epsilon}{16} f\alpha_T.
    \]
\end{enumerate}
\end{lemma}

\begin{proof}[Proof of The Gap Lemma, Parts 1 and 2]
\underline{Part 1.}
Fix a good process $i \in G$ and $t \in [T]$. 
For $r \in [m]$, let $\delta_r \in\{-1,0,1\}$ be 
the outcome of its $r$th 
coin-flip, being 0 if the adversary never lets it flip $r$ coins in iteration $t$.
Then for any $r<s$, $\E[\delta_r\delta_s] = 0$. This clearly holds when $\delta_s=0$,
and if the adversary lets the $s$th flip occur, $\E[\delta_r\delta_s \mid \delta_s\neq 0, \delta_r]=0$ since $\delta_s\in\{-1,1\}$ is uniform and independent of $\delta_r$.
Therefore, $\E[(X_i(t))^2] = \E[(\sum_{r=1}^m \delta_r)^2] = \sum_{r=1}^m \E[\delta_r^2] + \sum_{r \neq s} \E[\delta_r\delta_s] = \sum_{r=1}^m \E[\delta_r^2] \leq m$.

Now consider the sequence of random variables $(S_t)_{t\in [0,T]}$
where $S_0 = 0$ and $S_t = S_{t-1} + (X_i(t))^2 - m$. 
Since $\E[S_t \mid S_{t-1},\dots,S_0] \leq S_{t-1}$, 
$(S_t)$ is a supermartingale. 
For all $t \in [T]$ we guarantee $|X_i(t)| \leq \XMAX$,
so $|S_t - S_{t-1}| = |(X_i(t))^2 - m| \leq \XMAX^2$. 
Hence, by Azuma's inequality (\autoref{thm:azuma}, Appendix~\ref{sect:tailbounds}), $S_T \leq \XMAX^2\sqrt{T(c\ln n)}$ with probability $1-\exp\{-(\XMAX^2\sqrt{T(c\ln n)})^2/2T\XMAX^2\}=1-n^{-\Omega(c)}$. 
Therefore, with high probability, for all $i\in[n]$,
\[
\DEV(i) = \sum_{t=1}^T (w_i X_i(t))^2 = w_i^2(S_T + Tm) \leq 
w_i^2(Tm + \XMAX^2\sqrt{T(c\ln n)})
=
w_i^2 m\left(T + \sqrt{T(c\ln n)^3}\right)
=
w_i^2 \cdot\alpha_T.
\]
\underline{Part 2.}
Fix a $t\in[T]$ and let $\delta_{i,r} \in\{-1,0,1\}$ be 
the outcome of the $r$th coin-flip of $i$ in iteration $t$.  
By the same argument as above,
$\E[X_i(t)X_j(t)] = \E[(\sum_r \delta_{i,r})(\sum_s \delta_{j,s})] = \sum_{r,s} \E[\delta_{i,r}\delta_{j,s}] = 0$. 
Now consider the sequence $(S_t)_{t\in[0,T]}$ where 
$S_0 = 0$ and $S_t = S_{t-1} + X_i(t)X_j(t)$. 
It follows that 
$\E[S_t \mid S_{t-1},\dots,S_0] = S_{t-1}$,
so $(S_t)$ is a martingale. 
By assumption, for all $t$, 
both $|X_i(t)|,|X_j(t)| \leq \XMAX$. 
So, $|S_t - S_{t-1}| = |X_i(t)||X_j(t)| \leq \XMAX^2$. By Azuma's inequality (\autoref{thm:azuma}, Appendix~\ref{sect:tailbounds}), 
$S_T \leq \XMAX^2\sqrt{Tc\ln n}$ 
with probability $1-n^{-\Omega(c)}$.
Therefore, with high probability, for all $i,j$,
\[
\CORR(i,j) 
= \sum_{t=1}^T w_iw_jX_i(t)X_j(t) 
\leq w_iw_j \XMAX^2 \sqrt{Tc\ln n}
= w_iw_j\cdot m\sqrt{T(c\ln n)^3}
= w_iw_j\cdot \beta_T. \qedhere
\]
\end{proof}

Part 3 of the Gap Lemma is proved in 
Lemmas~\ref{lem:weight-bound}--\ref{lem:bad-process-corr-dev}.  
By Invariant~\ref{inv:weights}, the total 
weight loss of the good players is at most
the weight loss of the bad players plus $\epsilon^2 f/8$. 
Define $\rho$ to be the relative weight loss of the bad players:
\[
\mbox{$\rho\geq 0$ is such that } \sum_{i\in B} w_i = (1-\rho)f.
\]
Thus, at this moment $\sum_{i\in G} (1-w_i) \leq \rho f + \epsilon^2 f/8$.
Remember that the scheduling adversary can
allow the protocol to progress while neglecting
to schedule up to $f$ good players.  Thus,
in Lemma~\ref{lem:weight-bound} we consider
an arbitrary set $\hat{G}\subset G$ of $n-2f$ good players.

\begin{lemma} 
\label{lem:weight-bound}
If \autoref{inv:weights} holds then
\begin{enumerate}
    \item \label{lem:weight-bound-good} For any $\hat{G} \subseteq G$ with $|\hat{G}| = n-2f$, $\sum_{i \in \hat{G}} w_i^2 \geq (1-\max\{\rho/2, \epsilon/8\})^2(n-2f)$. 
    \item \label{lem:weight-bound-bad} $\sum_{(i\neq j)\in B^2} w_iw_j \leq (1-\rho)^2f^2$ and $(1-\rho)^2f \leq \sum_{i \in B} w_i^2 \leq (1-\rho)f$.
\end{enumerate}

\end{lemma}

\begin{proof}
We first claim that, for any real numbers $\hat{w}_1,\dots,\hat{w}_k \in [0,1]$, if $\sum_{i=1}^k \hat{w}_i = (1-\hat{\rho})k$ for some $\hat{\rho} \in [0,1]$, then $(1-\hat{\rho})^2 k \leq \sum_{i=1}^k \hat{w}_i^2 \leq (1-\hat{\rho})k$. 
The lower bound follows from 
Jensen's inequality (\autoref{thm:jensen}) 
and is achieved when all weights are equal.
The upper bound follows from the fact that $\hat{w}_i^2 \leq \hat{w}_i$.

\underline{Part 1.}
Note that
\begin{align*}
\textstyle 
\sum_{i \in \hat{G}} w_i &= 
\textstyle 
n-2f - \sum_{i \in \hat{G}} (1-w_i) 
\geq n - 2f-(\rho+\epsilon^2/8) f & \mbox{(\autoref{inv:weights})}\\
&= (1-\textstyle\frac{\rho+ \epsilon^2/8}{2+\epsilon})(n-2f) \\
&\ge (1 - \max\{\rho/2, \epsilon/8\})(n-2f)
\end{align*}
Thus the relative weight loss from $\hat{G}$'s point of view is less than $\hat{\rho}=\max\{\rho/2, \epsilon/8\}$, and 
from the first claim of the proof, 
$\sum_{i \in \hat{G}} w_i^2 \geq (1-\max\{\rho/2, \epsilon/8\})^2(n-2f)$. 

\underline{Part 2.} From the first claim of the proof
with $\hat{\rho}=\rho$, we have 
$(1-\rho)^2 f \leq \sum_{i \in B} w_i^2 \leq (1-\rho)f$.  For the other claim, 
\[
\textstyle
\sum_{(i \neq j) \in B^2} w_iw_j = (\sum_{i \in B} w_i)^2 - \sum_{i \in B} w_i^2 \leq (1-\rho)^2f^2 - (1-\rho)^2f \leq (1-\rho)^2f^2. \qedhere
\]
\end{proof}

Let us recall a few key facts about the game.
Before $\BB_t$ is constructed the adversary
commits to its desired direction $\sigma(t)$.
The $m\times n$ matrix 
$\BB_t$ is complete when it has $n-f$ full columns,
therefore the adversary \emph{must} allow
at least $m(n-2f)$ coins to be flipped by
good players.  We define $S_G(t)$ to be the 
weighted sum of all the coin flips flipped by good players.
I.e., if the set $G$ is stable throughout iteration $t$ then
\[
S_G(t) = \sum_{i\in G} w_i X_i(t).
\]
If a process $i$ were corrupted in the middle of iteration $t$ then
only a prefix of its coin flips would contribute to $S_G(t)$.
If $\sgn(S_G(t)) = \sigma(t)$
then the adversary is happy.
For example, it can just let the sum of the coin flips controlled by corrupted
players sum up to zero, which does not look particularly suspicious.  
However, if $\sgn(S_G(t))\neq\sigma(t)$
then the adversary must counteract the good coin flips.  
Due to disagreements in the state of the blackboard 
(see \autoref{lem:bad-process-large-coin}), players can disagree
about the sum of blackboard entries by up to 
$f$, so the adversary may only need to 
counteract the good players by $-\sigma(t)S_G(t)-f$.
\autoref{lem:real-blackboard-min-obj-inst}
lower bounds the second moment of this objective.

\begin{lemma}
\label{lem:real-blackboard-min-obj-inst}
For all $t \in [T]$,
\begin{align*}
\E[(\max\{0,-\sigma(t)S_G(t)-f\})^2] &\ge m\left((1-\max\{\rho/2, \epsilon/8\})^2(n/2-f)  - {\epsilon f/16}\right).
\end{align*}
\end{lemma}
\begin{proof}
Let $S_r$, $r\ge 0$, 
be the weighted sum of the first 
$m(n-2f) + r$ coin flips generated by good players,
and $Z_r = (\max\{0,-\sigma(t)S_r-f\})^2$
be the objective function for $S_r$.
The adversary can choose to stop letting
the good players flip coins at any time after $m(n-2f)$, thus $\E[(\max\{0,-\sigma(t)S_G(t)-f\})^2] = \E[Z_{2fm}]$, which we argue is at least $\E[Z_0]$.
Note that if $Z_{r-1}=0$ then
the adversary has achieved the minimum objective and has no interest in further flips,
so $\E[Z_r \mid Z_{r-1}=0]\geq Z_{r-1}$.
If $Z_{r-1}>0$, then 
were the adversary to allow some
$i\in G$ to flip another coin,
we would have 
$S_r = S_{r-1} + w_i\delta_r$, $\delta_r\in\{-1,1\}$,
and
\[
    Z_r = \begin{cases}
    (-\sigma(t)S_{r-1} - f + w_i)^2 = Z_{r-1} + 2w_i(-\sigma(t)S_{r-1} - f) + w_i^2 &\textrm{with probability } \frac{1}{2}, \\
    (-\sigma(t)S_{r-1} - f - w_i)^2 = Z_{r-1} - 2w_i(-\sigma(t)S_{r-1} - f) + w_i^2 &\textrm{with probability } \frac{1}{2}.
    \end{cases}
\]
Thus, $\E[Z_r \mid Z_{r-1} > 0, |\delta_r|>0] 
= Z_{r-1} + w_i^2 \geq Z_{r-1}$, 
i.e., if the adversary is trying to \emph{minimize} 
the objective function 
$(\max\{0,-\sigma(t)S_G(t)-f\})^2$,
it will not allow any good coin flips beyond the bare minimum.

To lower bound $\E[Z_0]$, the analysis above shows that
any adversary minimizing this objective will let the player $i$ 
with the smallest weight flip the next coin
(thereby minimizing $w_i^2$), conditioned on any prior history. 
Thus, in the worst
case the $n-2f$ good players with the smallest weights each flip $m$ coins.

We compute $\E[Z_0]$ under this strategy. 
Since $S_0 = \sum_{i=1}^{n-2f} \sum_{r=1}^m w_i\delta_{i,r}$, 
where $\delta_{i,r} \in \{-1,1\}$ are fair coin flips, 
$\Pr(-\sigma(t)S_0 \geq 0) \geq \frac{1}{2}$ 
by a simple bijection argument ($\delta_{i,r} \mapsto -\delta_{i,r}$). 
Hence, $\E[Z_0] \geq 
\frac{1}{2}\E[Z_0 \mid -\sigma(t)S_0 \geq 0]$. 
Continuing,
\begin{align*}
\textstyle 
\E[Z_0 \mid -\sigma(t)S_0 \geq 0] 
&= \E[(-\sigma(t)S_0-f)^2 \mid -\sigma(t)S_0 \geq 0] + \E[Z_0 - (-\sigma(t)S_0-f)^2 \mid -\sigma(t)S_0 \geq 0] \\
&\geq \E[(-\sigma(t)S_0-f)^2 \mid -\sigma(t)S_0 \geq 0] -f^2\\
&\geq \E[(S_0)^2 \mid -\sigma(t)S_0 \geq 0] - 2f\E[-\sigma(t)S_0 \mid -\sigma(t)S_0 \geq 0]\\ 
&= \E[(S_0)^2] - 2f\E[|S_0|].
\end{align*}
The first inequality comes from the fact that $Z_0 \neq (-\sigma(t)S_0 - f)^2$ only when $-\sigma(t)S_0 \in [0,f)$ (given the conditioning) and in this range
is $-(-\sigma(t)S_0 - f)^2 \geq -f^2$.
The second inequality comes from expanding $(-\sigma(t)S_0-f)^2$, linearity of expectation, 
and the fact that $\sigma(t)^2=1$.
Since $\E[\delta_{i,r}\delta_{i',r'}] = 0$ for $(i,r) \neq (i',r')$,  
\begin{align}\label{eqn:S0-sqaured}
\textstyle 
 \E[(S_0)^2] = \E[(\sum_{i,r} w_i\delta_{i,r})^2] = \sum_{i,r,i',r'} w_iw_{i'}\E[ \delta_{i,r}\delta_{i',r'}] = m\sum_{i=1}^{n-2f} w_i^2 \, .
\end{align}
We bound the expected value of $|S_0|$ as follows
\begin{align*}
\E[|S_0|] 
&\leq \sqrt{\E[(S_0)^2]}
&(\operatorname{Var}[|S_0|] = \E[(S_0)^2]- (\E[|S_0|])^2 \geq 0)\\
&=\sqrt{m\sum_{i=1}^{n-2f} w_i^2} \le \sqrt{mn} & (\autoref{eqn:S0-sqaured})\\
\E[(S_0)^2] &\ge m(1-\max\{\rho/2, \epsilon/8\})^2(n-2f).
&(\autoref{lem:weight-bound})
\intertext{and putting it all together we have}
\E[Z_0] 
&\geq 
 \frac{1}{2}\E[Z_0 \mid -\sigma(t)S_0 \geq 0]\\
&\geq \frac{1}{2}\left(\E[(S_0)^2]-2f\cdot \E[|S_0|]\right)\\
&\geq \frac{1}{2}\left(m(1-\max\{\rho/2, \epsilon/8\})^2(n-2f) - 2f\sqrt{mn}\right)\\
&\geq m\left((1-\max\{\rho/2, \epsilon/8\})^2(n/2-f) - {\epsilon f/16}\right)
\end{align*}
The last line follows since $m=\Theta(n/\epsilon^2)$. 
\end{proof}

\begin{lemma}
\label{lem:sum-squared-bound}
With high probability, 
for every $t \in [T]$, 
$\max\{0,-\sigma(t)S_G(t)-f\}^2 \leq cmn\ln n$.
\end{lemma}
\begin{proof}
The total number of good coin 
flips is at most $mn$.
By a Chernoff-Hoeffding bound (\autoref{thm:hoeffding}, Appendix~\ref{sect:tailbounds}), $S_G(t) \leq \sqrt{cmn\ln n}$ with high probability and the lemma follows.
\end{proof}

\begin{lemma}
\label{lem:gap-lemma-lb}
With high probability,
\[
\sum_{t=1}^T \max\{0, -\sigma(t)S_G(t) - f\}^2 \ge
m\left[\left((1-\max\{\rho/2, \epsilon/8\})^2(n/2-f) - {\epsilon f/16}\right)T - n\sqrt{T(c\ln n)^3}\right].
\]
\end{lemma}

\begin{proof}
Let $\gamma = ((1 - \max\{\rho/2,\epsilon/8\})^2(n/2-f) - {\epsilon f/16})$.
Consider the sequence of random variables $A_0,A_1,\dots,A_T$, where $A_0=0$ and $A_{t} = A_{t-1} + \max\{0, -\sigma(t)S_G(t)-f\}^2 - m\gamma$. By \autoref{lem:real-blackboard-min-obj-inst}, $\E[A_t \mid A_{t-1},\dots,A_0] \geq 0$. So, $(A_t)$ is a submartingale. By \autoref{lem:sum-squared-bound}, with high probability, for all $t \in [T]$, 
$\max\{0,-\sigma(t)S_G(t)-f\}^2 \leq m\gamma'$, where 
$\gamma' = cn\ln n$. 
Assuming this holds, $|A_t - A_{t-1}| \leq m\gamma'$ and, by Azuma's inequality (\autoref{thm:azuma}, Appendix~\ref{sect:tailbounds}), $A_T \leq -m\gamma'\sqrt{Tc\ln n}$
with probability $1-n^{-\Omega(c)}$. 
Therefore, with high probability, \[
\sum_{t=1}^T \max \{0, -\sigma(t)S_G(t) - f\}^2 = m\gamma T+ A_T  
\geq m(\gamma T - \gamma'\sqrt{Tc\ln n}). \qedhere
\]
\end{proof}

\begin{lemma}
\label{lem:bad-process-large-coin}
For every epoch in which no players are corrupted,
\[
\sum_{i \in B} \DEV(i) + \sum_{(i \neq j) \in B^2} \CORR(i,j) \ge \sum_{t=1}^T \max\{0, -\sigma(t)S_G(t) - f\}^2 \, .
\]
\end{lemma}

\begin{proof}
Define $S_B(t)$ to be the sum of coin flips declared by corrupted players.  
I.e., if $B$ were stable throughout iteration $t$ then 
$S_B(t) = \sum_{i \in B} w_iX_i(t)$. 
Then
\[
    \sum_{t\in [T]} (S_B(t))^2 = 
    \sum_{t\in [T]}\left(\sum_{i \in B} (w_iX_i(t))^2 + \sum_{(i \neq j) \in B^2} w_iw_jX_i(t)X_j(t)\right) 
    = \sum_{i \in B} \DEV(i) + \sum_{(i \neq j) \in B^2} \CORR(i,j).
\]
In iteration $t\in [T]$,
the adversary must 
convince at least one good process $p$
that $\sgn\left(\sum_i w_iX_i^{(p)}(t)\right) = \sigma(t)$.
By \autoref{thm:blackboard}, 
$\sum_i|X_i^{(p)}(t) - X_i(t)|\leq f$
and hence the total disagreement between 
$p$'s weighted sum and 
the true weighted sum is
\[
\sum_i\left| w_i X_i^{(p)}(t) - w_iX_i(t)\right| = \sum_i w_i\left| X_i^{(p)}(t) - X_i(t)\right|\leq f.
\]
Thus, if $-\sigma(t)S_G(t) \geq f$ 
(the good players sum
is in the non-adversarial direction by at least $f$) the bad players
must correct it by setting $\sigma(t)S_B(t) \geq -\sigma(t)S_B(t)-f$.  
Therefore, for any $S_G(t)$
we must have 
$(S_B(t))^2 \geq \max\{0,-\sigma(t)S_G(t)-f\}^2$
and the lemma follows.
\end{proof}

Recall from Parts 1 and 2 of The Gap Lemma
(\autoref{lem:gap-lemma}) that every good
player $i\in G$ has $\DEV(i) \leq w_i^2\alpha_T$ and every good pair 
$(i,j)\in G^2$ has $\CORR(i,j) \leq w_iw_j\beta_T$.  \autoref{lem:bad-process-corr-dev} lower bounds the \emph{excess} of the $\DEV/\CORR$-values involving bad players, beyond these allowable thresholds.

\begin{lemma}
\label{lem:bad-process-corr-dev}
In any epoch in which no processes are corrupted, 
With high probability,
\[
    \sum_{i\in B} \max\{0,  \DEV(i) - w_i^2\alpha_T\} + \sum_{(i\neq j)\in B^2} \max\{0, \CORR(i, j) - w_iw_j\beta_T\} \ge \frac{\epsilon}{16} f\alpha_T \, .
    \]
\end{lemma}

\begin{proof}
By \autoref{lem:gap-lemma-lb} and
\autoref{lem:bad-process-large-coin}, with high probability,
\begin{align}\label{eqn:dev-plus-corr}
    \sum_{i\in B} \DEV(i) + \sum_{(i\neq j)\in B^2} \CORR(i, j) 
    \ge m\left(\left((1-\max\{\rho/2, \epsilon/8\})^2(n/2-f) - {\epsilon f/16}\right)T - n\sqrt{T(c\ln n)^3}\right).
\end{align}
Recall that
$\alpha_T = m(T+\sqrt{T(c\ln n)^3})$, $\beta_T = m\sqrt{T(c\ln n)^3}$, 
and, by
\autoref{lem:weight-bound}, that
$\sum_{i \in B} w_i^2 \leq (1-\rho)f$ and $\sum_{(i \neq j) \in B^2} w_iw_j \leq (1-\rho)^2 f^2$.  Putting these together
we have
\begin{align}\label{eqn:alpha-plus-beta}
    \alpha_T\sum_{i \in B} w_i^2 + \beta_T\sum_{(i \neq j) \in B^2} w_iw_j 
    &\leq m\left(T + \sqrt{T(c\ln n)^3}\right)\cdot (1-\rho)f + m\sqrt{T(c\ln n)^3}\cdot (1-\rho)^2f^2.
\end{align}
The expression we wish to bound is at least (\ref{eqn:dev-plus-corr}) minus (\ref{eqn:alpha-plus-beta}), namely:
\begin{align}\label{eqn:simplify-gap}
&m\left[\Big((1-\max\{\rho/2, \epsilon/8\})^2(n/2-f) - {\epsilon f/16} - (1-\rho)f\Big)T - \left((1-\rho)f + (1-\rho)^2f^2\right)\sqrt{T(c\ln n)^3}\right]
\end{align}
Now depending on the larger value of $\rho/2$ and $\epsilon/8$, there are two cases expanding \autoref{eqn:simplify-gap}.

\paragraph{Case 1: $\rho/2 \le \epsilon/8$.} In this case, we simplify \autoref{eqn:simplify-gap} by setting $\rho= 0$.
\begin{align*}
(\mbox{\autoref{eqn:simplify-gap}}) & \ge mT \left[ (1-\epsilon/8)^2(1+\epsilon/2)f - \epsilon f/16 - f\right] - f(f+1)m\sqrt{T(c\ln n)^3} \\
&\ge mT\left[(1+\epsilon/4 - \epsilon^2/8)f - \epsilon f/16 - f\right] - n^2m\sqrt{T(c\ln n)^3}\\
&\ge \frac{\epsilon}{8}fmT - n^2m\sqrt{T(c\ln n)^3} \tag{$\epsilon \le 1/2$}\\
&\ge \frac{\epsilon}{16}f\alpha_T.
\end{align*}

\paragraph{Case 2: $\rho/2 > \epsilon/8$.} In this case, we expand the $(1-\rho/2)^2$ term and simplify \autoref{eqn:simplify-gap} using the
identity $n=(4+\epsilon)f$.
\begin{align*}
(\mbox{\autoref{eqn:simplify-gap}}) &\geq mT\left[(n/2 - 2f)(1-\rho) + (n/2-f)\rho^2/4 - {\epsilon f/16}\right] - n^2m\sqrt{T(c\ln n)^3}\\
&= mTf\left[(\epsilon/2)(1-\rho) + (1+\epsilon/2)\rho^2/4 - \epsilon/16\right] - n^2m\sqrt{T(c\ln n)^3},\\
\intertext{which is minimized when $\rho =\epsilon/(1+\epsilon/2)$, hence}
&\geq mTf\left[\frac{\epsilon}{2}\left(1-\frac{\epsilon}{1+\epsilon/2}\right) + (1+\epsilon/2)\left(\frac{\epsilon}{1+\epsilon/2}\right)^2/4 - {\epsilon /16}\right] - n^2m\sqrt{T(c\ln n)^3}\\
&= mTf\left[\frac{7\epsilon}{16} - \frac{\epsilon^2}{4(1+\epsilon/2)}\right] - n^2m\sqrt{T(c\ln n)^3},
\intertext{and since $T=\Theta(n^2\ln^3 n/\epsilon^2)$ and
$\alpha_T = m(T+\sqrt{T(c\ln n)^3})$, with $\epsilon < 1/2$ this is lower bounded by}
&\geq \frac{\epsilon}{4}f\alpha_T. \qedhere
\end{align*}
\end{proof}

\begin{remark}
We are able to upper bound correlation scores between two good players, 
and lower bound the average correlation score between two bad players.
However, the correlations between good and bad players cannot be usefully limited.  This is why  \autoref{lem:bad-process-large-coin} and
\autoref{lem:bad-process-corr-dev} 
only apply to epochs in which no processes are corrupted, since any $\CORR(i,j)$ score is difficult to analyze when $i$ is corrupted halfway through the epoch.
\end{remark}

\subsection{Weight Updates}\label{sect:weight-update}

\newcommand{\UNIT}{\ensuremath{\mathit{unit}}}
\newcommand{\WeightUpdate}{\ensuremath{\mathsf{Weight\text{-}Update}}}
\newcommand{\RisingTide}{\ensuremath{\mathsf{Rising\text{-}Tide}}}

When the $T$ iterations of an epoch $k$ are complete, 
we reduce the weight vector $(w_i)$ in preparation for epoch $k+1$.
According to The Gap Lemma, if an individual deviation score $\DEV(i)$ is too large, 
$i$ is bad w.h.p., and if a correlation score $\CORR(i,j)$ is too large, 
$B\cap\{i,j\}\neq\emptyset$ w.h.p., so reducing \emph{both} $i$ and $j$'s weights
by the \emph{same} amount preserves \autoref{inv:weights}.  
With this end in mind, \WeightUpdate{} 
(\autoref{alg:weights-update-real-bb}) 
constructs a complete, vertex- and edge-capacitated graph $G$ on $[n]$, 
finds a fractional maximal matching $\mu$ in $G$, 
then docks the weights of $i$ and $j$ by $\mu(i,j)$, 
for each edge $(i,j)$.

\begin{definition}[Fractional Maximal Matching]
Let $G=(V,E,c_V,c_E)$ be a graph where $c_V : V\rightarrow \mathbb{R}_{\geq 0}$ are vertex capacities
and $c_E : E\rightarrow \mathbb{R}_{\geq 0}$ are edge capacities.  
A function $\mu : E\rightarrow \mathbb{R}_{\geq 0}$
is a \emph{feasible fractional matching} 
if $\mu(i,j)\leq c_E(i,j)$ and $\sum_{j} \mu(i,j) \leq c_V(i)$.
It is \emph{maximal} if it is not strictly dominated by any feasible $\mu'$. 
The \emph{saturation level} of $i$ is $\sum_{j} \mu(i,j)$; 
it is \emph{saturated} if this equals $c_V(i)$.
An edge $(i,j)$ is \emph{saturated} if $\mu(i,j)=c_E(i,j)$.
(Note that contrary to convention, a self-loop $(i,i)$ only
counts \emph{once} against the capacity of $i$, not twice.)
\end{definition}

\paragraph{Rounding Weights Down.} Recall that if $p$ participates in a blackboard $\BB_t$, 
that every other process can compute the weight vector computed from $p$'s local view $\BB^{(p,t-1)}$ 
through blackboard $t-1$.  The processes use a unified weight vector in which $w_i$ is derived 
only from $i$'s local view:
\[
w_i = \left\{\begin{array}{ll}
w_i^{(i)}  & \mbox{ if $w_i^{(i)} > w_{\min} = \frac{\sqrt{n\ln n}}{T}$}\\
0           & \mbox{ otherwise.}
\end{array}\right.
\]

As we will see, the maximum pointwise disagreement  $|w_i^{(p)}-w_i^{(q)}|$ 
between processes $p,q$
is at most $w_{\min}$, and as a consequence, if any $p$ thinks $w_i^{(p)}=0$ then all processes agree that $w_i=0$.

\newcommand{\ErrT}{\ensuremath{err_T}}

\paragraph{Excess Graph.} The \emph{excess graph} $G=(V, E, c_V,c_E)$ used in \autoref{alg:weights-update-real-bb} 
is a complete undirected graph on $V=[n]$, 
including self-loops, capacitated as follows:
\begin{align*}
    c_V(i) &= w_i, \\
    c_E(i, i) &= \frac{16}{\epsilon f\alpha_T} \cdot \max\{0, \DEV(i)-w_i^2\alpha_T \}, \\
    c_E(i, j) &= \frac{16}{\epsilon f\alpha_T} \cdot 2\max\{0, \CORR(i, j)-w_iw_j\beta_T \},
\end{align*}

\medskip

The reason for the coefficient of ``2'' in the definition of $c_E(i,j)$ is that
$(i,j)$ is a single, undirected edge, but it represents two correlation scores 
$\CORR(i,j)=\CORR(j,i)$, which were accounted for separately 
in \autoref{lem:bad-process-corr-dev}.
By parts 1 and 2 of The Gap Lemma, $c_E(i,j)=0$ whenever both $i$ and $j$ are good.

\begin{algorithm}[h]
\caption{\WeightUpdate{} from the perspective of process $p$.}\label{alg:weights-update-real-bb} 
\textbf{Output:} Weights $(w_{i,k})_{i\in [n], k\geq 0}$ where $w_{i,k-1}$ refers to the
weight $w_i$ \emph{after} processing epoch $k-1$, and is used throughout epoch $k$.\\
\begin{algorithmic}[1]
\State Set $w_{i, 0}\gets 1$ for all $i$.\Comment{All weights are 1 in epoch 1.}
\For{epoch $k=1, 2, \ldots, \KMAX$}\Comment{$\KMAX = $ last epoch}
\State Play the coin flipping game for $T$ iterations with weights $(w_{i,k-1})$ and let 
$\DEV^{(p)}$ and $\CORR^{(p)}$ be the resulting deviation and correlation scores known to $p$.
Construct the excess graph $G_k^{(p)}$ with capacities:
\begin{align*}
c_{V}(i) &= w_{i, k-1}, \\[4pt]
c_{E}(i, i)  &=  \displaystyle\frac{16}{\epsilon f\alpha_T}\cdot \max\left\{0, \DEV^{(p)}(i) - w_{i, k-1}^2\alpha_T\right\},\\[8pt]
c_{E}(i, j) &=  \displaystyle\frac{16}{\epsilon f\alpha_T}\cdot 2\max\left\{ 0, \CORR^{(p)}(i, j) - w_{i, k-1}w_{j, k-1}\beta_T\right\}.
\end{align*}
\State $\mu_k \gets \RisingTide(G_k)$ \Comment{A maximal fractional matching}
\State For each $i$ set \[
w_{i, k}^{(p)} \gets w_{i, k-1} - \sum_{j} \mu_k(i, j).{\hspace{3.3cm}}{\ }
\]
\State Once $(w_{i,k}^{(q)})$ are known for $q\in [n]$, set
\[
w_{i,k} = \left\{\begin{array}{ll}
w_{i,k}^{(i)}  & \mbox{ if $w_{i,k}^{(i)} > w_{\min} \bydef \frac{\sqrt{n\ln n}}{T}$}\\
0           & \mbox{ otherwise.}
\end{array}\right.
\]
\EndFor
\end{algorithmic}
\end{algorithm}

The \WeightUpdate{} algorithm from the perspective of process $p$ is presented in \autoref{alg:weights-update-real-bb}.
We want to ensure that the fractional matchings computed by good processes are numerically very close to each other, 
and for this reason, we use a specific maximal matching
algorithm called \RisingTide{} (\autoref{alg:rising-tide})
that has a continuous Lipschitz property, i.e.,
small perturbations to its input yield bounded perturbations to its output.
Other natural maximal matching algorithms such as \emph{greedy} do not have this property.

\subsubsection{Rising Tide Algorithm}
\label{sect:rising-tide}

The \RisingTide{} algorithm initializes $\mu=0$ and 
simply simulates
the continuous process of increasing 
all $\mu(i,j)$-values in lockstep, so long as 
$i$, $j$, and $(i,j)$ are not saturated.
At the moment one becomes saturated, $\mu(i,j)$ is 
frozen at its current value.

\begin{algorithm}[h]
\caption{$\RisingTide(G=(V,E,c_V,c_E))$}\label{alg:rising-tide}
\begin{algorithmic}[1]
\State $E'\gets \{(i, j)\in E\ |\ c_E(i, j) > 0\}$.
\State $\mu(i, j)\gets 0$ for all $i,j\in V$.
\While{$E' \neq\emptyset$}
    \State Let $\mu_{E'}(i, j) = \begin{cases} 1 & \text{ if } (i, j)\in E'\\ 0 & \text{ otherwise.}\end{cases}$.
    \State Choose maximum $\epsilon > 0$ such that $\mu'=\mu+\epsilon \mu_{E'}$ is a feasible fractional matching.
    \State Set $\mu\gets \mu'$.
    \State $E' \gets E' - \{(i,j) \mid \mbox{$i$ or $j$ or $(i,j)$ is saturated}\}$\Comment{$\mu(i,j)$ cannot increase}\label{line:remove-edges}
\EndWhile
\State \Return $\mu$.
\end{algorithmic}
\end{algorithm}

\begin{lemma}\label{lem:rising-tide-correctness}
\RisingTide{} (\autoref{alg:rising-tide}) correctly 
returns a maximal fractional matching.
\end{lemma}

\begin{proof}
Obvious.
\end{proof}

Recall that $c_V(i)$ is initialized to be the (old) weight $w_i$
and the new weight is set to be $c_V(i)-\sum_j \mu(i,j)$.  We are mainly
interested in differences in the new weight vector computed by processes
that start from slightly different graphs $G,H$.  
\autoref{lem:rising-tide-output} bounds these output differences
in therms of their input differences.

\begin{lemma}[Rising Tide Output]\label{lem:rising-tide-output}
Let $G=(V, E, c_V^G, c_E^G)$ and $H=(V, E, c_V^H, c_E^H)$ be two capacitated graphs, 
which differ by $\eta_E = \sum_{i, j} |c_E^G(i, j) - c_E^H(i, j)|$ 
in their edge capacities
and 
$\eta_V = \sum_{i} |c_V^G(i) - c_V^H(i)|$ 
in their vertex capacities.
Let $\mu_G$ and $\mu_H$ be the fractional matching computed by \RisingTide{} (\autoref{alg:rising-tide}) 
on $G$ and $H$ respectively.
Then: 
\[
\sum_{i} \left| \left(c_V^G(i) -  \sum_j \mu_G(i, j)\right) - \left(c_V^H(i) - \sum_j \mu_H(i, j)\right)\right| \le \eta_V + 2\eta_E.
\]
\end{lemma}

See Appendix~\ref{section:rising-tide-proofs} for proof
of \autoref{lem:rising-tide-output}.

\subsection{Error Accumulation and Reaching Agreement}
\label{sect:bounding-error}

\newcommand{\etaP}{\ensuremath{\eta^{(p)}}}

The maximum number of epochs is $\KMAX=2.5f$.
Let $k\in [1,\KMAX]$ be the index of the current epoch, 
and let $w_{i, k-1}$ be the weights that were used in the execution of $\CoinFlip()$ during epoch $k$.
Upon completing epoch $k$, each process $p$ 
applies \autoref{alg:weights-update-real-bb} 
to update the consensus weight vector $(w_{i, k-1})_{i\in [n]}$ 
to produce a local weight vector $(w_{i,k}^{(p)})_{i\in [n]}$,
and then the consensus weight vector $(w_{i,k})_{i\in [n]}$ 
used throughout epoch $k+1$.

\begin{lemma}[Maintaining \autoref{inv:weights}]\label{lem:maintaining-inv}
Suppose for some $\epsilon>0$ that $n=(4+\epsilon)f$, 
$m=\Theta(n/\epsilon^2)$, and $T=\Theta(n^2\log^3 n/\epsilon^2)$.
At any point in epoch $k\in[1,\KMAX]$, with high probability,
\[
\sum_{i\in G} (1 - w_{i, k-1}) \le \sum_{i\in B} (1 - w_{i, k-1}) + \frac{\epsilon^2}{\sqrt{n}}\cdot (k-1).
\]
\end{lemma}

\begin{proof}
We prove by induction on $k$.
For the base case $k=1$ all the weights are $1$ so \autoref{lem:maintaining-inv} clearly holds.
We will now prove that if the claim holds for $k$, it holds for $k+1$ as well.
Fix any good process $p$.  The vector $(w_{i,k}^{(p)})$ is derived from $(w_{i,k-1})$ 
by deducting at least as much weight from bad processes as from good processes, with high probability,
and $(w_{i,k})$ is derived from $(w_{i,k}^{(q)})_{q\in [n], i\in [n]}$ by setting $w_{i,k}=w_{i,k}^{(i)}$
and rounding down to 0 if it is at most $w_{\min}$.  
Thus, by the inductive hypothesis,

\begin{align*}
    \sum_{i\in G} (1 - w_{i, k}^{(p)}) \le \sum_{i\in B} (1 - w_{i, k}^{(p)}) + \frac{\epsilon^2}{\sqrt{n}}\cdot (k-1) 
\end{align*}

Therefore,

\begin{align*}
\sum_{i\in G} (1 - w_{i, k}) &\le \sum_{i\in B} (1 - w_{i, k}) + \frac{\epsilon^2}{\sqrt{n}}\cdot (k-1) + \sum_{i\in [n]} |w_{i, k}^{(p)} - w_{i, k}^{(i)}| + w_{\min{}} n_0,
\end{align*}
where $n_0$ is the number of processes whose weight is rounded down to $0$ after epoch $k$.

Hence, it suffices to show that $\sum_{i\in [n]} |w_{i, k+1}^{(p)} - w_{i, k+1}^{(i)}| + w_{\min{}}n_0\le \epsilon^2/\sqrt{n}$.
By \autoref{lem:rising-tide-output}, the computed weight difference between process $p$ and any process $q$ 
can be bounded by twice the sum of all edge capacity differences.
According to \autoref{alg:weights-update-real-bb}, the edge capacities differ 
due to underlying disagreement on the $\DEV(i)$ and $\CORR(i, j)$ values. Thus,

\begin{align*}
|w_{q, k}^{(p)} - w_{q, k}^{(q)}| &\le 2\cdot \frac{16}{\epsilon f\alpha_T}\left(\sum_i \left|\DEV^{(p)}(i) - \DEV^{(q)}(i)\right| + \sum_{i\neq j} \left| \CORR^{(p)}(i, j) - \CORR^{(q)}(i, j) \right|\right)\\
\intertext{
By \autoref{thm:blackboard}, two processes may only disagree in up to $f$ cells of the blackboards
$(\BB_1,\ldots,\BB_t)$.  Since the sum of each column in each blackboard is bounded by $\XMAX$,
we have $|\DEV^{(p)}(i)-\DEV^{(q)}(i)| < 2\XMAX$ for at most $f$ values of $i$,
and $|\PCORR(i, j)-\CORR^{(q)}(i, j)| < 2\XMAX$ 
for at most $nf$ pairs $i\in B, j\in(G\cup B)$.
Continuing,}
&\le 2\cdot \frac{16}{\epsilon f\alpha_T}\bigg(f\cdot 2\XMAX + nf\cdot 2\XMAX\bigg)
\\
&\le \frac{64 (n+1)\XMAX}{\epsilon mT} \tag{$\alpha_T\ge mT$}\\
&\le \frac{ \sqrt{n\ln n}}{T}\tag{using $m=\Omega(n/\epsilon^2)$}\\
&=w_{\min{}}
\end{align*}
Now the inductive step for $k$ holds by noticing that
\begin{align*}
\sum_{i\in [n]}|w_{i, k}^{(p)} - w_{i, k}^{(i)}|  + w_{\min{}}n_0 &\le  2w_{\min{}}n\\
&\le \frac{\epsilon^2}{n^{1.5}\log^{2.5} n}\cdot n  \tag{using $T=\Omega(n^2\log^3 n/\epsilon^2)$}\\
&< \frac{\epsilon^2}{\sqrt{n}}.
\end{align*}
Therefore, with $\KMAX=2.5f$ we obtain \autoref{inv:weights}. That is, for any weight vector $(w_i)$ that are used on a blackboard,
\begin{align*}
\sum_{i\in G}(1-w_i) &\le \sum_{i\in B}(1 - w_i) + \frac{\epsilon^2}{\sqrt{n}} \cdot 3f\\
&\le \sum_{i\in B}(1 - w_i) + \frac18 \epsilon^2 f. \tag{whenever $n\ge 576$}\\
\end{align*}

Note that Invariant~\autoref{inv:weights} is also preserved whenever a process is corrupted, transferring it from $G$ to $B$.
\end{proof}

The next observation and \autoref{lem:weight-update-progress} shows that the weight 
of every bad process becomes $0$ after running $\KMAX$ epochs of \WeightUpdate{}s 
without reaching agreement.

\begin{observation}\label{obs:weights-closed-to-zero}
For any $i$ and $k$, if there exists process $p$ such that $w_{i, k}^{(p)}=0$, then $w_{i, k} = 0$.
\end{observation}

\begin{proof}
In the proof of \autoref{lem:maintaining-inv} it was shown that 
$|w_{i, k}^{(p)} - w_{i, k}^{(i)}|\le \sqrt{n\ln n}/T = w_{\min{}}$,
hence if $w_{i,k}^{(p)}=0$, $w_{i,k}$ is rounded down to 0. 
See \autoref{alg:weights-update-real-bb}.
\end{proof}

\begin{lemma}
\label{lem:weight-update-progress}
If agreement has not been reached after $\KMAX=2.5f$ epochs, 
all bad processes have weight $0$, with high probability.
\end{lemma}

\begin{proof}
There are at most $f$ epochs in which the adversary corrupts at least one process.
We argue below that after all other epochs, in the call to 
\WeightUpdate, the total edge capacity of the graph induced by $B$ is at least 1.
This implies that in each iteration of \WeightUpdate,
either some $i\in B$ with $c_V(i)=w_i>w_{\min{}}$ becomes
saturated (and thereafter $w_i=0$ by \autoref{obs:weights-closed-to-zero}), 
or the total weight of all processes in $B$ 
drops by at least 2. 
The first case can occur at most $f$ 
times and the second at most $f/2$, 
hence after $\KMAX=2.5f$ epochs, 
all bad players' weights are zero, with high probability.

We now prove that the total edge capacity is at least 1.  Recall that each edge $(i,j)$, $i\neq j$, represents the two correlation scores 
$\CORR(i,j)$ and $\CORR(j,i)$. 
Hence, by \autoref{lem:bad-process-corr-dev}, 
the sum of edge capacities on $B$ is:
\begin{align*}
\sum_{\{i, j\}\subset B} c_E(i, j)
    &= \frac{16}{\epsilon f\alpha_T}\left( \sum_{i\in B} \max\{0, \DEV(i) - w^2_{i, k} \alpha_T \}
    + \sum_{(i\neq j)\in B^2} \max\{0, \CORR(i, j) - w_{i, k}w_{j, k}\beta_T\} \right) \\
    &\ge \frac{16}{\epsilon f\alpha_T} \left( \frac{\epsilon}{16} f\alpha_T\right)  \tag{by \autoref{lem:gap-lemma}}\\
    &\ge 1.\qedhere
\end{align*}
\end{proof}

\begin{lemma}\label{lem:good-processes-agree}
Suppose \autoref{inv:weights} holds. 
In any iteration in which the bad processes have zero weights, 
the good processes agree on the outcome of the coin flip, with constant probability.
\end{lemma}

\begin{proof}
Let $S=\sum_i w_iX_i(t)$ be the weighted sum of the players.
Through its scheduling power, 
the adversary may still be able to create disagreements between good players on the outcome
of the coin-flip if $S\in [-f,f]$. 
Moreover, good process still possess $\Omega(n)$ total weight by \autoref{inv:weights}.
With constant probability, 
$|S|$ is larger than its standard deviation, namely $\Theta(\sqrt{mn})$, 
which is much larger than $f$ as $m = \Omega(n/\epsilon^2)$.  Thus, with constant
probability all good players agree on the outcome.
\end{proof}

\begin{theorem}
Suppose $n=(4+\epsilon)f$ where $\epsilon>0$, $m=\Theta(n/\epsilon^2)$, and 
$T=\Theta(n^2\log^3 n/\epsilon^2)$.
Using the implementation of $\CoinFlip()$ from \autoref{sect:iterated-coin-flipping},
\BrachaAgreement{} solves Byzantine agreement with 
probability 1 in the full information, asynchronous model against an adaptive adversary. 
In expectation the total communication time is $\tilde{O}((n/\epsilon)^4)$.
The local computation at each process is polynomial in $n$.
\end{theorem}

\begin{proof}
By \autoref{lem:weight-update-progress}, after $\KMAX=2.5f$ epochs, 
all bad processes' weights become zero, with high probability.
From then on, by \autoref{lem:good-processes-agree}, 
each iteration of \BrachaAgreement{} achieves agreement with constant probability.
Thus, after one more epoch, all processes reach agreement with high probability.
The total communication time (longest chain of dependent messages) 
is $O((\KMAX+1)mT) = \tilde{O}((n/\epsilon)^4)$.  
If, by chance, the processes fail to reach agreement after this much time,
they restart the algorithm with all weights $w_i=1$ and try again.  Thus,
the algorithm terminates with probability 1.
\end{proof}

\bibliographystyle{alpha}
\bibliography{main}

\appendix

\section{Proofs from \autoref{sect:preliminaries}}\label{sect:appendix-prelim-proofs}

\subsection{Reliable Broadcast: Proof of \autoref{thm:reliable-broadcast}}\label{sect:proof-reliable-broadcast}

\reliablebroadcastthm*

\begin{proof}
According to the first line of \ReliableBroadcast, no message 
$m_{p,\ell}$ can be accepted until $m_{p,\ell-1}$ is accepted, if $\ell>1$.
This establishes the FIFO property.  The other correctness properties
follow from several claims.

We claim that if two good processes $q,q'$ send $(\operatorname{ready},m_{p,\ell})$
and $(\operatorname{ready},m'_{p,\ell})$, then $m_{p,\ell}=m'_{p,\ell}$.
Suppose not. Let $q$ be the first good process to send a 
$(\operatorname{ready},m_{p,\ell})$
message, and let $q'$ be the first good process to send
a $(\operatorname{ready},m'_{p,\ell})$, for some $m'_{p,\ell}\neq m_{p,\ell}$. 
By definition, $q$ and $q'$
received strictly more than $(n+f)/2$ $(\operatorname{echo},m_{p,\ell})$
and 
$(\operatorname{echo},m'_{p,\ell})$ messages, respectively.
Thus, at least $2\ceil{(n+f+1)/2}-n \geq f+1$ processes sent both $q$ and $q'$
conflicting $(\operatorname{echo},\cdot)$ messages, and therefore some 
\emph{good} process sent conflicting $(\operatorname{echo},\cdot)$ 
messages, which is impossible.

We now claim that if a good process $q$ accepts $m_{p,\ell}$
then every good process eventually accepts $m_{p,\ell}$.
It follows that $q$ has already accepted $m_{p,1},\ldots,m_{p,\ell-1}$.
By induction, every other good process eventually
accepts $m_{p,\ell-1}$.  Before accepting $m_{p,\ell}$, 
$q$ received at least $f+1$ 
$(\operatorname{ready},m_{p,\ell})$ messages from 
good processes.  These $f+1$ messages will eventually be delivered
to all $n-f\geq 2f+1$ good processes, causing all to send their own
$(\operatorname{ready},m_{p,\ell})$ messages and
eventually accept the same value.

If the sender $p$ is good, then every good process will clearly eventually
accept $m_{p,\ell}$.  Moreover, as a consequence of the claims above,
if $p$ is bad it is impossible for good processes to accept different
messages $m_{p,\ell}\neq m'_{p,\ell}$.
\end{proof}

\subsection{Iterated Blackboard}\label{sect:appendix-blackboard}

The $\IteratedBlackboard$ algorithm  
uses the 
reliable broadcast primitives \emph{broadcast} and \emph{accept} (\autoref{thm:reliable-broadcast}) to construct a series
of blackboards $\BB = (\BB_1,\BB_2,\ldots)$, the columns of which
are indexed by process IDs in $[n]$ and the rows of which are indices
in $[0,m]$.  The blackboard proper consists of rows $1,\ldots,m$; 
the purpose of row zero is to reduce disagreements between the 
views of good processes.
Every process $p$ maintains $\BB^{(p)} = (\BB_1^{(p)},\BB_2^{(p)},\ldots)$, 
where $\BB_t^{(p)}(r,i)$ records
the value written by process $i$ to $\BB_t(r,i)$
and accepted by process $p$, 
or $\bot$ if no such value has yet been accepted by $p$.  Each process maintains
a vector $\last^{(p)}$ indicating the position of the last accepted write
from each process, i.e., $\last^{(p)}(i)=(t,r)$ if $p$ accepted $i$'s write
to $\BB_t(r,i)$, but has yet to accept 
any subsequent writes from $i$ to 
$\BB_t$, nor to $\BB_{t+1},\BB_{t+2},\ldots$.

\autoref{alg:blackboard} gives the algorithm \IteratedBlackboard$(t)$
for generating $\BB_t$ from the perspective of process $p$.  
Process $p$ may only begin executing it if $t=1$ or if it has finished 
executing \IteratedBlackboard$(t-1)$ and therefore 
already fixed $\BB^{(p,t-1)}$ and $\maxlast^{(p)}_{t-1}$.

\begin{algorithm}
    \caption{\IteratedBlackboard$(t)$ \emph{ from the perspective of process $p$}}
    \label{alg:blackboard}
    \begin{algorithmic}[1]
            \State Set $\complete(t) \gets \mathit{false}$ 
            and set $\zeta \gets \maxlast^{(p)}_{t-1}$ if $t>1$ or any dummy value if $t=1$.
            \emph{Broadcast} the write $\BB_t(0,p) \gets \zeta$.
            \vskip 5pt \label{line:write-first}
                
            \State \textbf{upon} \emph{validating} $\geq n-f$ $\ack(\BB_t(m,q))$'s for $\geq n-f$ different $q$ for the first time: 
            {\par \enspace} 
            set $\complete(t) \gets \mathit{true}$ and
            $\last_t^{(p)} \gets \last^{(p)}$, then
            \emph{broadcast} the vector $\last^{(p)}_t$.
            \vskip 5pt \label{line:until}
        
        
        \State{\textbf{upon} \emph{validating}  $\ack(\BB_t(r,p))$'s from $\geq n-f$ different processes for the first time: {\par \enspace}
         \textbf{if} $\neg \complete(t) \wedge (r < m)$ \textbf{then} \emph{generate} a value $\zeta$ and \emph{broadcast} the write $\BB_t(r+1,p)\gets \zeta$.} \vskip 5pt \label{line:increment-j}

        \State{\textbf{upon} \emph{validating} $\BB_t(r,q)$ from process $q$ for the first time: {\par \enspace}
        set $\pBB_t(r,q) \gets \BB_t(r,q)$ and $\last^{(p)}(q) \gets (t,r)$; 
        \textbf{if} $\neg \complete(t)$ \textbf{then} \emph{broadcast} $\ack(\BB_t(r,q))$.} \vskip 5pt \label{line:ack}

        \State \textbf{upon} \emph{validating} $\last^{(q)}_t$ vectors 
        from $\geq n-f$ different processes $q$ for the first time:
            {\par \enspace}
            set $\maxlast^{(p)}_t(i) \gets \max_{q} \{\last^{(q)}_t(i)\}$ (point-wise maximum, lexicographically).
            {\par \enspace}
            At this point $\BB^{(p,t)}=\left(\BB_1^{(p,t)},\ldots,\BB_t^{(p,t)}\right)$ is fixed as follows:
            \[
            \BB_{t'}^{(p,t)}(r,i) = \left\{
            \begin{array}{ll}
            \BB_{t'}^{(p)}(r,i) & \mbox{if $(t',r)\leq \maxlast_t^{(p)}(i)$ and $r\in [1,m]$}\\
            \bot        & \mbox{otherwise} 
            \end{array}
            \right.
            \]
            \vskip 5pt
            \label{line:maxlast}
    \end{algorithmic}
\end{algorithm}

In Line~\ref{line:write-first}, $p$ broadcasts the write
$\BB_t(0,p)\gets \maxlast^{(p)}_{t-1}$ (or a dummy value if $t=1$).
This serves two purposes: first, to let other processes know that $p$
has begun $\IteratedBlackboard(t)$, and second, to let them know
exactly how it fixed the history $\BB^{(p,t-1)}$ after $\BB_{t-1}$
was completed.
Once any process accepts this message, 
if they do not already consider $\BB_t$ to be \emph{complete}, 
then they broadcast an \emph{acknowledgement} (Line~\ref{line:ack}). 
In general, once process $p$ receives $n-f$
acknowledgements for its write to $\BB_t(r,p)$,
if $r<m$ and $p$ does not consider $\BB_t$ complete,
it proceeds to generate and broadcast a write to
$\BB_t(r+1,p)$ (Line~\ref{line:increment-j}). 
A process $q$ generates an acknowledgement $\ack(\BB_t(r,p))$ 
only if $q$ accepts $\BB_t(r,p)$ 
and does not consider $\BB_t$ to be complete (Line~\ref{line:ack}).

An important point is that the conditions of the ``\textbf{upon}'' statements on Lines~\ref{line:until}, \ref{line:increment-j}, \ref{line:ack}, and \ref{line:maxlast} are checked \emph{whenever} process $p$ accepts a new message, in the order that they are written. 
In particular, process $p$ may set 
$\BB_t^{(p)}(r,q)$ after it considers $\BB_t$
to be complete, and even after it has
moved on to the execution of $\IteratedBlackboard(t+1)$.
(These can be thought of as \emph{retroactive} corrections to $\BB_t$.)
However, the body of each ``\textbf{upon}'' 
statement is executed at most once.

Eventually, process $p$ will see that a set of at least $n-f$ columns of $\BB_t$ are full, i.e., for each such column $q$, it has received at least $n-f$ acknowledgements for the $m$'th row of column $q$. Once this occurs, process $p$ sets $\complete(t)$ to true (Line~\ref{line:until}). 
At this point there can still be considerable disagreements 
between $p$'s local view and another process's 
local view of the blackboard. To (mostly) resolve this,
process $p$ broadcasts the current state of its last-vector 
$\last^{(p)} = \last^{(p)}_t$.
It then waits to receive $\last^{(q)}_t$ vectors from at least $n-f$ different $q$ before it finalizes what it considers to be the last position of each column $i$ of $\BB$ at the end of the $t$'th iteration, i.e., $\maxlast^{(p)}_t(i) = \max_q \{\last^{(q)}_t(i)\}$ (Line~\ref{line:maxlast}). 

A critical aspect of the protocol is that $p$ refrains 
from participating in any broadcast unless it has 
validated the message, i.e., accepted messages that
are a prerequisite for its existence.  Specifically:
\begin{enumerate}
    \item No process participates in the broadcast of a write to
        $\BB_t(r,p)$ 
        unless $(t,r)=(1,0)$ (it's $p$'s first write)
        or it has accepted the last write from $p$, 
        which is $\BB_t(r-1,p)$ or if $r=0$, some $\BB_{t-1}(r',p)$.  
         (This is already captured by the FIFO property 
         of \autoref{thm:reliable-broadcast} but it is 
         useful to highlight it again.)
         Moreover, no process participates in the broadcast 
         of a write $\BB_t(0,p)\gets \maxlast^{(p)}_{t-1}$
         unless it has accepted $n-f$ $\last^{(q)}_{t-1}$ vectors
         whose point-wise maxima are exactly $\maxlast^{(p)}_{t-1}$.
    \item No process participates in the broadcast of $\ack(\BB_t(r,q))$ 
    until it has accepted $\BB_t(r,q)$.
    \item No process participates in the broadcast of $\BB_t(r+1,q)$ until 
    it has accepted acknowledgements for $\BB_t(r,q)$ from at least $n-f$ processes.
    \item No process participates in the broadcast of 
    $\last^{(q)}_t$ until it has,
    for all columns $i$ and $\last^{(q)}_t(i) = (t',r)$ already accepted $\BB_{t'}(r,i)$.
    In other words, the process must first accept every blackboard value that $q$ purports to have accepted.
\end{enumerate}
To emulate this, we modify the \ReliableBroadcast{} 
implementation (\autoref{alg:reliable-broadcast})
so that when process $p$ \emph{receives} a 
message inviting it to participate in a \ReliableBroadcast{} 
(e.g.,~an ``init'' message), 
it simply delays reacting to the message until the prerequisite 
conditions are met. 
In the next lemma, we show that these changes to the broadcast 
mechanism do not cause any deadlocks.

\begin{lemma}
If at least $n-f$ good processes execute $\IteratedBlackboard(t)$, then every good process that executes $\IteratedBlackboard(t)$ eventually sets $\maxlast^{(p)}_t$ and $\BB^{(p,t)}$.
\end{lemma}

\begin{proof}
First, we claim that, if any good process considers $\BB_t$ complete, then every good process that executes $\IteratedBlackboard(t)$ eventually considers $\BB_t$ complete. Indeed, if a good process considers $\BB_t$ complete, then it must have accepted $\ack(\BB_t(m,q))$s from at least $n-f$ processes, for at least $n-f$ values of $q$ (Line~\ref{line:until}). 
By the properties of reliable broadcast (\autoref{thm:reliable-broadcast}), 
every other good process eventually accepts these acknowledgements as well. Thus, every good process that executes $\IteratedBlackboard(t)$ eventually considers $\BB_t$ complete.

Next, we claim that, if at least $n-f$ good processes consider $\BB_t$ complete, then every good process $p$ that executes $\IteratedBlackboard(t)$ will eventually set $\maxlast^{(p)}_t$. Indeed, by Line~\ref{line:until}, every good process $q$ that considers $\BB_t$ complete will broadcast a $\last^{(q)}_t$ vector. By the properties of reliable broadcast, any blackboard values accepted by $q$ will eventually be accepted by every good process and, hence, every good process will eventually participate in $q$'s broadcast of $\last^{(q)}_t$. Thus, every good process $p$ eventually accepts $\last^{(q)}_t$ vectors from at least $n-f$ different processes $q$ and sets $\maxlast^{(p)}_t$ (Line~\ref{line:maxlast}).

Finally, since at least $n-f$ good processes execute $\IteratedBlackboard(t)$, by our preceding discussion, it suffices to show that at least one such process considers $\BB_t$ complete. Suppose, for a contradiction, that this is not the case. Consider any good process $p$ that executes $\IteratedBlackboard(t)$ with the minimum number of writes to its column of $\BB_t$.

Suppose $p$ writes to row $m$. Then, by minimality, every good process $q$ that executes $\IteratedBlackboard(t)$ writes to row $m$ in their respective columns, i.e., broadcasts $\BB_t(m,q)$. By the properties of reliable broadcast, the $n-f$ $\ack(\BB_t(m-1,q))$'s that allow each such process $q$ to broadcast $\BB_t(m,q)$ will eventually be accepted by every good process. Thus, every good process will eventually accept $\BB_t(m,q)$ and, as they do not consider $\BB_t$ complete by assumption, they will broadcast $\ack(\BB_t(m,q))$. Therefore, every good process will receive at least $n-f$ $\ack(\BB_t(m,q))$ for at least $n-f$ different $q$ and, consequently, consider $\BB_t$ complete (Line~\ref{line:until}), which is a contradiction.

Now suppose $p$ writes to row $r < m$. If $r > 0$, then the $n-f$ $\ack(\BB_t(r-1,p))$'s that allow $p$ to broadcast $\BB_t(r,p)$ will eventually be accepted by every good process. Hence, every good process will eventually participate in $p$' broadcast of $\BB_t(r,p)$ and accept $\BB_t(r,p)$. Similarly, if $r = 0$. Since $p$ does not write to row $r+1 \leq m$, it never accepts $n-f$ $\ack(\BB_t(r,p))$'s. Since at least $n-f$ good process accept $\BB_t(r,p)$, it follows that at least one such process does not broadcast an $\ack(\BB_t(r,p))$. By Line~\ref{line:ack}, this process must have considered $\BB_t$ to be complete by the time it accepts $\BB_t(r,p)$, which is a contradiction.

Therefore, in both cases, we reach the desired contradiction.
\end{proof}

Recall that $p$'s view of the history after executing $\IteratedBlackboard(t)$ is $\BB^{(p,t)}$, defined to be:
 \[
            \BB_{t'}^{(p,t)}(r,i) = \left\{
            \begin{array}{ll}
            \BB_{t'}^{(p)}(r,i) & \mbox{if $(t',r)\leq \maxlast_t^{(p)}(i)$ and $r\in [1,m]$}\\
            \bot        & \mbox{otherwise} 
            \end{array}
            \right.
\]
In other words, we obtain $\BB^{(p,t)}_{t'}$ 
by stripping off the zeroth row of 
every $\BB^{(p)}_{t'}$ matrix
and replacing any values in column $i$ after $\maxlast^{(p)}_t(i)$ 
with $\perp$. 
We emphasize that, in contrast to the local 
matrix $\BB^{(p)}_t$ of process $p$, 
once $\BB^{(p,t)}_{t'}$ is set, it never changes.
(In the context of Bracha's algorithm, $p$ uses $\BB^{(p,t)}_t$ 
to decide the outcome of the $t$th call to $\CoinFlip()$, 
but the series of blackboards $\BB^{(p,T)}$ is used to decide how 
$p$ reduces the weight vector $(w_i)$ after the first epoch.)

\begin{lemma}\label{lem:nminusf-acks}
Suppose at least $n-f$ good processes execute $\IteratedBlackboard(t)$ and some good process $q$ accepts $\ack(\BB_t(r,i))$ from at least $n-f$ different processes.
Then every good process $p$ that finishes iteration $t$ has
$\BB^{(p)}_t(r,i) = \BB^{(p,t)}_{t}(r,i) = \BB_t(r,i)$.
\end{lemma}

\begin{proof}
Since there are at most $f$ bad processes, $q$ accepts $\ack(\BB_t(r,i))$'s from at least $n-2f$ good processes.  Thus, $\last^{(q')}_t(i) \geq (t,r)$
holds for a set $S_0$ of at least $n-2f$ processes $q'$.
Similarly, when a good process $p$ finishes iteration $t$, it has received $\last^{(q')}_t$ vectors from a set $S_1$ of at least $n-2f$ good processes.
Since $n>3f$, $n-2f > (n-f)/2$ and $S_0\cap S_1$ contains at least one common good process. Thus, $\maxlast^{(p)}_t(i) = \max_{q'} \{\last^{(q')}_t(i)\} \geq (t,r)$,
meaning that $p$ will not finish Line~\ref{line:maxlast} until
it accepts $\BB_t(r,i)$ (due to validation), recording it in $\pBB_t$ and hence
$\BB^{(p,t)}_{t}$.
\end{proof}

\begin{lemma}\label{lem:bb-disagreement}
Suppose that each good process executes 
$\IteratedBlackboard(1),\ldots,\IteratedBlackboard(t)$,
beginning iteration $t'+1$ only after it has executed Line~\ref{line:maxlast} 
of iteration $t'$.
Then, for any two good processes $p,q$ that finish iteration $t$, 
$\BB^{(p,t)}$ and $\BB^{(q,t)}$ disagree in at most $f$ 
positions \emph{in total}.
If they disagree on the contents of any position, one is $\perp$.
\end{lemma}

\begin{proof}
The properties of reliable broadcast ensures that $\BB^{(p)}_t,\BB^{(q)}_t$
cannot contain \emph{distinct} non-$\perp$ values in any position.  
Therefore we must argue that they differ in at most $f$ positions.
Fix a process $k$ and let $\BB_{t_{k}}(r_{k},k)$ be the \emph{last}
of $k$'s blackboard writes for which it accepted at least $n-f$ 
$\ack(\BB_{t_{k}}(r_{k},k))$s.
By \autoref{lem:nminusf-acks}, 
$\BB^{(p,t)}_{t_k}(r_k,k) = \BB^{(q,t)}_{t_k}(r_k,k)$. 
Moreover, due to validation, $p,q$ have both
accepted all of $k$'s blackboard writes prior to $\BB_{t_k}(r_k,k)$.
Subsequent blackboard writes of $k$ that could appear in the local matrices of $p$ and $q$ 
are limited to $\BB_{t_k}(r_k+1,k)$ (if $r_k < m$) and 
$\BB_{t_k+1}(0,k),\ldots,\BB_t(0,k)$ (if $t_k < t$). 
This follows
by assumption on $(t_k,r_k)$: at the first time when both 
$p$ and $q$ finish iteration $t$, $k$ has not accepted sufficiently 
many acknowledgements to attempt any of these writes. 
Since we strip off the zeroth rows of each local view to form
$\BB^{(p)}_{t},\BB^{(q)}_{t}$, 
they may only disagree in column $k$ at $\BB_{t_k}(r_k+1,k)$. 
Now, for at least $n-f$ processes $k$ we 
have $(t_k,r_k)=(t,m)$. 
Since row $j_k+1=m+1$ does not exist in any $\BB$-matrix, 
$\BB^{(p,t)},\BB^{(q,t)}$ may only disagree in $f$ cells in total.
\end{proof}

\begin{lemma}
If $\BB_{t+1}(1,q)\neq \bot$, then by the time $p$ fixes $\BB^{(p,t+1)}$, 
it is aware of $q$'s history $\BB^{(q,t)}$ through blackboard $t$.
\end{lemma}

\begin{proof}
Before $q$ wrote anything to $\BB_{t+1}(1,q)$ it must have written
$\BB_{t+1}(0,q)\gets \maxlast^{(q)}_{t}$ and caused $n-f$ 
acknowledgements $\ack(\BB_{t+1}(0,q))$ to be broadcast.
By \autoref{lem:nminusf-acks} every process will accept $q$'s write to 
$\BB_{t+1}(0,q)$ before fixing $\BB^{(p,t+1)}$, and hence be able
to reconstruct $\BB^{(q,t)}$ from $\maxlast^{(q)}_t$.
\end{proof}

In conclusion, we have proved the following theorem:

\blackboardthm*

In the context of our implementation of Bracha's algorithm, 
the purpose of \autoref{thm:blackboard}(3) is to ensure that 
if $q$ writes anything to $\BB_{t+1}$ (specifically $\BB_{t+1}(1,q)\neq \bot$)
that every other process $p$ can determine $w_q$, 
which is a function of $q$'s history $\BB^{(q,t)}$.

\section{Concentration Bounds}\label{sect:tailbounds}

Random variables $X_1,\dots,X_n$ are \emph{mutually independent} (or, just \emph{independent}) if for any sequence of real numbers $x_1,\dots,x_n$, the events $\{X_1 \leq x_1\},\dots,\{X_n \leq x_n\}$ are independent.

\begin{theorem} [Chernoff-Hoeffding bound]
\label{thm:hoeffding}
Let $X_1,\dots,X_n$ be independent random variables and let $X = \sum_{i=1}^n X_i$. If each $X_i \in [a_i,b_i]$ for some $a_i \leq b_i$, then, for all $t \geq 0$,
\[
\Pr\{ X - \mathbb{E}[X] \geq t \}, \Pr\{ X - \mathbb{E}[X] \leq -t \} \leq \exp\left(-\frac{2t^2}{\sum_{i=1}^n (b_i-a_i)^2}\right) \, .
\]
\end{theorem}

A sequence of random variables $X_0,X_1,\dots$ is a \emph{submartingale} if $\E[X_t \mid X_{t-1},\dots,X_0] \geq X_{t-1}$ for all $t$; it is a \emph{supermartingale} if $\E[X_t \mid X_{t-1},\dots,X_0] \leq X_{t-1}$; and it is a \emph{martingale} if it is both a submartingale and a supermartginale.

\begin{theorem} [Azuma's inequality]
\label{thm:azuma}
Let $X_0,X_1,\dots$ be a martingale such that, for each $i \geq 1$, $|X_i - X_{i-1}| \leq c_i$ for some $c_i \geq 0$. Then for all $n \geq 1$ and $t \geq 0$, 
\[
    \Pr\{X_n - X_0 \geq t\}, \Pr\{X_n - X_0 \leq -t \} \leq \exp\left(-\frac{2t^2}{\sum_{i=1}^n c_i^2}\right) \, .
\]
Moreover, the bound on $\Pr\{X_n - X_0 \geq t\}$ holds when $\{X_i\}_{i\geq 0}$ is a supermartingale while the bound on $\Pr\{X_n - X_0 \leq -t\}$ holds when $\{X_i\}_{i\geq 0}$ is a submartingale.
\end{theorem}

A function $f : \mathbb{R} \to \mathbb{R}$ is \emph{convex} if for any $x,y \in \mathbb{R}$ and $t \in [0,1]$, $f((1-t)x + ty) \leq (1-t)f(x) + tf(y)$.

\begin{theorem} [Jensen's inequality]
\label{thm:jensen}
Let $f : \mathbb{R} \to \mathbb{R}$ be convex. Then, for any random variable $X$,
\[
    f\left(\E[X]\right) \leq \E[f(X)]\, .
\]
\end{theorem}

\section{Rising Tide Algorithm Proofs}\label{section:rising-tide-proofs}

The goal to this entire section is to prove \autoref{lem:rising-tide-output}, the heart of bounding different perspectives from each process.

\subsection{Dependency Graphs}

We first introduce an idea of a \emph{dependency graph} that captures the 
moments when vertices become saturated in \autoref{alg:rising-tide}.
We will then use structural properties of dependency 
graphs to prove~\autoref{lem:rising-tide-output}.

\begin{definition}[Dependency Graph]
Let $D_G$ be a \emph{directed} graph with the same set of vertices $V(D_G)=V$.
Consider the execution of \autoref{alg:rising-tide}
on $G$.
For each edge $e=(i, j)\in E$, if at the moment $e$ is removed from the working set $E'$ (Line~\ref{line:remove-edges}), 
$i$ (resp. $j$) is saturated, then we include in $D_G$ a directed edge $j\to i$ (resp. $i\to j$).
Notice that if both $i$ and $j$ are saturated simultaneously,
then $D_G$ includes both edges $i\to j$ and $j\to i$.
\end{definition}

We first state a useful continuity property of \RisingTide, that if we continuously deform 
the input capacities, the output fractional matching also changes continuously.

\begin{lemma}[The Continuity Lemma]\label{lem:continuity-lemma}
Let $G$ and $H$ be two fractional matching instances where every vertex- 
and edge-capacity differs by at most $\xi$.
Then, for every edge $e$, $|\mu_G(e) - \mu_H(e)| \le F(n)\xi$ for some function $F$ 
which depends only on the size of the graph but not on $\xi$.
\end{lemma}

\begin{proof}
Without loss of generality we can assume
that each edge capacity $c_E^G(i,j) \leq \min\{c_V^G(i),c_V^G(j)\}$
is always bounded by the capacities of its endpoints.

Imagine running \RisingTide{} simultaneously on both $G$ and $H$,
stopping at the first saturation event that occurs in, say, 
$G$ but not $H$. (A ``saturation event'' is the saturation of a vertex or edge with \emph{non-zero} capacity.)
Let $\mu_G',\mu_H'$ be the fractional matchings at this time 
and $G',H'$
be the residual graphs, 
i.e., obtain new capacities by subtracting each $\mu_G'(i,j)$ from $c_V^G(i),c_V^G(j),$ and $c_E^G(i,j)$.
The maximum difference in vertex- or edge-capacities 
between $G',H'$ is $n\xi$.  
The argument can be applied inductively to $G',H'$, 
and since there are $O(n^2)$ saturation events, the maximum
difference between any capacity (and hence an $\mu$-value) 
is always bounded by $F(n)\xi$, where $F(n)=n^{O(n^2)}$.
\end{proof}

Note that the magnitude of $F$ is immaterial, so long as it depends only on $n$.
\autoref{lem:continuity-lemma} allows us to make several simplifying assumptions.
\begin{itemize}
    \item[A1.] First, although we are comparing
two graphs $G,H$ with possibly many capacity differences, we can assume w.l.o.g. that
they differ in precisely one vertex- or edge-capacity.
    \item[A2.] Second, we can assume that the dependency
graphs for $G$ and $H$ are identical.
    \item[A3.] Third, we can assume, via infinitesimal perturbations, that no two 
    vertices are saturated simultaneously.  In particular, this implies that $D_G$ is acyclic.  (See~\autoref{lem:rising-tide-output-properties}.)
\end{itemize}

\begin{lemma}[Basic Properties between $\mu$ and $D_G$]\label{lem:rising-tide-output-properties}
Assume graph $G$ satisfies assumption (A3).
Let $\mu$ be the output of $G$ from \autoref{alg:rising-tide}. Then: 
\begin{enumerate}
    \item[(1)] For any two edges $e_1$ and $e_2\in E$, if $e_1$ gets removed from $E'$ before $e_2$, 
    then $\mu(e_1) < \mu(e_2)$.
    \item[(2)] For each $u\in V$, all edges directed towards $u$ in $D_G$ have the same $\mu$-value.
    \item[(3)] For any edge $u\to v$ in $D_G$ and any edge $(v, w)\in E$, $\mu(u, v)\ge \mu(v, w)$.
    \item[(4)] (Monotonic Path Property) For any walk $u_0 \to u_1\to u_2\to \cdots $ on $D_G$, their $\mu$ values must be non-increasing. That is, $\mu(u_0, u_1)\ge \mu(u_1, u_2) \ge \cdots$.
    \item[(5)] (Directed Acyclic Graph Property) $D_G$ is a DAG.
\end{enumerate}
\end{lemma}

\begin{proof}
To show (1), it suffices to observe that in \autoref{alg:rising-tide} the fractional matching $\mu$ grows strictly increasing at each iteration.

\medskip

To show (2), it suffices to show that for each vertex $u\in V$ with any two incoming edges $v\to u$ and $w\to u$ on $D_G$,
$\mu(v, u) = \mu(w, u)$.
Suppose conversely and without loss of generality $\mu(v, u) > \mu(w, u)$.
By the time $(w, u)$ gets removed from the working set $E'$, $u$ is already saturated.
However, it is now impossible to increase $\mu(v, u)$ anymore, contradicting to the assumption that $\mu(v, u) > \mu(w, u)$.

\medskip

To show (3), we notice that at the time $(u, v)$ is removed from $E'$, $v$ is saturated.
At this moment any edge $(v, w)\in E$ incident to $v$ cannot increase its $\mu$ value anymore.
Hence, $(v, w)$ will be removed from $E'$ at the same time with $(u, v)$ or prior to the time when $(u, v)$ is removed from $E'$.
Thus, by (1) we have $\mu(u, v)\ge \mu(v, w)$. (4) follows from (3) directly.

\medskip

To show (5), assume contradictory that there exists a cycle $u_0\to u_1\to\cdots \to u_0$ in $D_G$. By the monotonic path property (4), all edges $\mu(u_i, u_{i+1})$ have the same fractional value when they were removed from the working set $E'$ in the rising tide algorithm. Moreover, by definition of $D_G$, all vertices are simultaneously saturated, which contradicts (A3).
\end{proof}

Henceforth (A3) is assumed to hold in all graphs.

\subsection{Proof of~\autoref{lem:rising-tide-output}}

To prove~\autoref{lem:rising-tide-output}, it suffices to show (via an interpolating argument) that the statement holds whenever (1) exactly one saturated vertex changes capacity but no edge changes capacity, or (2) exactly one saturated edge changes capacity but no vertex changes capacity.
Moreover, with the continuity lemma (\autoref{lem:continuity-lemma}) it suffices to prove the statements (1) and (2) with the assumption that $D_G=D_H$.


We start with some observations when there is only one change on the capacities between $G$ and $H$.

\begin{lemma}\label{observation:self-saturated-edges}
Let $G$ and $H$ be two input graphs with the same dependency graph $D := D_G=D_H$.
If $(i, j)$ is an edge for which neither
$i\to j$ nor $j\to i$ appear in $D$, then $\mu_G(i, j) = \mu_H(i, j)$.
\end{lemma}

\begin{proof}
From the definition of the dependency graphs,
if both $i$ and $j$ are not saturated by the time $(i, j)$ gets removed from $E'$,
then by the fact (\autoref{lem:rising-tide-correctness}) that both $\mu_G$ and $\mu_H$ are maximal fractional matchings,
$\mu_G(i, j) = c_E^G(i, j) = c_E^H(i, j) = \mu_H(i, j)$.
\end{proof}

Suppose graph $G$ and $H$ have the same capacities except at some vertex $s\in V$. Then,
\autoref{observation:self-saturated-edges} implies that if we run the rising tide algorithm on both instances $G$ and $H$, the first moment they differ from each other, must be the case where on one graph $s$ is saturated but on another graph $s$ is not.
In this case, we can think of $s$ being the \emph{source} of all the disagreement.
Intuitively, if we look at an edge $e$ where $\mu_G(e)\neq \mu_H(e)$, we should be able to trace and blame this disagreement to the source of shenanigans. 

\begin{lemma}\label{observation:path-to-changing-vertex}
Assume that $G$ and $H$ differ only in the capacity of one vertex $s$ and that $D := D_G=D_H$.
Consider any edge $(i, j)$ such that $\mu_G(i, j)\neq \mu_H(i, j)$. Then,
there exists a (possibly empty) paths in the dependency graph $D$ from $i$ and $j$ to $s$.
Moreover, any edge $e$ on this path satisfies $\mu_G(e)\neq \mu_H(e)$.
\end{lemma}

\begin{proof}
Without loss of generality, when we consider an edge $(i, j)$ with $\mu_G(i, j)\neq \mu_H(i, j)$, we may always assume $j\to i$ appears in $D_G$. (This edge must exist by \autoref{observation:self-saturated-edges}.)
That is, when $(i, j)$ is removed from the \RisingTide{} algorithm that runs on $G$ it is because $i$ is saturated.
Since $D_G=D_H$, $i$ is also saturated by the time when $(i, j)$ gets removed on both instances $G$ and $H$.
Now we prove this lemma by induction on all edges from the smallest $\mu_G$ value to the largest $\mu_G$ value.

\medskip

\noindent\textbf{Base Case.}
Suppose $(i, j)$ is one of the edges with the minimum $\mu_G$-value such that $\mu_G(i, j)\neq \mu_H(i, j)$.
Since this is the first moment when the algorithm behaves differently, and we assume that $j\to i$ on $G$, 
it follows that at time $\mu_G(i, j)$, the vertex $i$ is saturated in $G$ but not in $H$. 
Moreover, all other edges incident to $i$ have the same $\mu_G$-value at this time.
Therefore $c_V^G(i) < c_V^H(i)$, and hence $i=s$.  There is a trivial path from $i=s$ to $s$ and a path
from $j$ to $s$ via $j\to i$.

\medskip

\noindent\textbf{Inductive Case.}
Now let us prove the inductive case. Suppose $\mu_G(i, j)\neq \mu_H(i, j)$ and when 
$(i, j)$ is removed from $E'$, the vertex $i$ is saturated.  If $i=s$ then we are done. 
Otherwise, we have $c_V^G(i, j) = c_V^H(i, j)$.
By \autoref{lem:rising-tide-output-properties} statements (2) and (3) and summing up all fractional matching values around the vertex $i$,
we know that there exists an edge $(i, j')$ with 
$\mu_G(i, j')\neq \mu_G(i, j)$ and also $\mu_G(i, j') \neq \mu_H(i, j')$.
By \autoref{lem:rising-tide-output-properties} statement (1) we know that $\mu_G(i, j') < \mu_G(i, j)$.
By the induction hypothesis and \autoref{observation:self-saturated-edges}, 
we know that $i\to j'$ in $D$ and there must be a path from $j'$ to $s$ on $D$. 
Therefore, there exists paths from $i$ and $j$ to $s$ in $D$ as well.
\end{proof}

Now, we prove the simplest version of \autoref{lem:rising-tide-output}
where only one vertex capacity is different with the assumption that the dependency graphs are the same.

\begin{lemma}\label{lem:rising-tide-diff-by-one-vertex}
Assume $G$ and $H$ only differ in the capacity of one vertex $s$, and that $D := D_G=D_H$.
Then, the total differences among the remaining vertex capacities can be bounded by
\[
\sum_{i} \left| \left(c_V^G(i) -  \sum_j \mu_G(i, j)\right) - \left(c_V^H(i) - \sum_j \mu_H(i, j)\right)\right|\le |c_V^G(s) - c_V^H(s)|.
\]
\end{lemma}

\newcommand{\COEF}{\nu}
\newcommand{\DDIFF}{D_{\mathrm{diff}}}
\begin{proof}
By \autoref{observation:path-to-changing-vertex}, 
all edges that have different fractional matching values 
 form a subgraph $\DDIFF$ of $D$ with $s$ being the only minimal element.
If $s$ is not saturated then there are no incoming edges to $s$. 
By \autoref{observation:path-to-changing-vertex} 
we know that $\DDIFF=\emptyset \implies \mu_G=\mu_H$ 
and in this case the equality holds for the statement.

\medskip

Observe that whenever there is an incoming edge to a vertex $i$ in $D$, the vertex $i$ must be saturated.
Since we are measuring differences in the remaining vertex capacities, the only place where such disagreement could happen is on all maximal vertices of $\DDIFF$.
Let $T$ be the set of maximal vertices, i.e., those without incoming edges.

\medskip

We prove a certain inequality by induction over 
all sets $S$ such that $S\subseteq V-T$ and $S$ is \emph{downward closed}, meaning there is no outgoing edge from $S$ to $V-S$.  As a consequence $s\in S$.
Let $\partial S$ be the set of incoming edges
from $V-S$ to $S$.  
We will prove that for any coefficients
$\{\COEF_{i\to j} \in [-1, 1]\}_{(i\to j)\in \partial S}$  
we have
\[
\left|\sum_{(i\to j)\in \partial S} \COEF_{i\to j} (\mu_G(i, j) - \mu_H(i, j)) \right| \le |c_V^G(s) - c_V^H(s)|.
\]

\medskip

\paragraph{Base Case.} 
The minimal downward closed set is $S=\{s\}$.
By \autoref{lem:rising-tide-output-properties} statement (2) all incoming edges have the same $\mu_G(i, s)-\mu_H(i, s)$ values.
That is, all terms in $\{\mu_G(i, s) - \mu_H(i, s)\}$ are of the same sign and hence claim is true for the base case.

\medskip

\paragraph{Inductive Case.}
Consider any downward-closed set $S\subseteq V-T$ 
with $|S|\geq 2$, and let $\{\COEF_{i\to j}\in [-1, 1]\}$ 
be any set of coefficients on the fringe $\partial S$.
Let $u\neq s$ be any maximal element in $S$.

Let $X_{\operatorname{in}}$ 
and $X_{\operatorname{out}}$
be the set of 
incoming and outgoing edges incident to $u$.
Since $S$ is downward-closed, we have
\[
\partial S = \partial (S - \{u\}) \cup X_{\operatorname{in}} - X_{\operatorname{out}}.
\]

Now, by \autoref{lem:rising-tide-output-properties} we know that each incoming edge $(i\to u)$ in $X_{\operatorname{in}}$ has the same fractional matching value in both $\mu_G$ and $\mu_H$. We denote the difference by  $\Delta\bydef \mu_G(i, u)-\mu_H(i, u)$.

Let $\COEF_u = \frac{1}{|X_{\operatorname{in}}|}\left(\sum_{(i\to u)\in X_{\operatorname{in}}}\COEF_{i\to u}\right) \in[-1, 1]$ be the average coefficient among all incoming edges.
Since $u$ is saturated, we have
\begin{align*}
&\ \ \ \sum_{(u\to j)\in X_{out}} \COEF_u(\mu_G(u, j) - \mu_H(u, j))
+
\sum_{(i\to u)\in X_{in}} \COEF_{i\to u} (\mu_G(i, u) - \mu_H(i, u)) \\
&= \sum_{(u\to j)\in X_{out}} \COEF_u(\mu_G(u, j) - \mu_H(u, j))
+
\COEF_u |X_{in}|\cdot \Delta \tag{by definition of $\COEF_u$}\\
&=\COEF_u \left(\sum_{(u\to j)\in X_{out}} (\mu_G(u, j) - \mu_H(u, j))
+
\sum_{(i\to u)\in X_{in}} (\mu_G(i, u) - \mu_H(i, u))
\right)
\\
&= \nu_u(c_V^G(u) - c_V^H(u)) \\
&= 0.
\end{align*}

Now, by removing $u$ from $S$ we have obtained a smaller subset which we can apply induction hypothesis on.
Define coefficients $\{\COEF'_{i\to j}\}$ with $\COEF'_{u\to j} = -\nu_u$ for all $(u\to j)\in X_{out}$ and $\COEF'_{i\to u} = \COEF_{i\to j}$ for all unrelated edges not incident to $u$.
Then, we have
\begin{align*}
    \left|\sum_{(i\to j)\in \partial S} \COEF_{i\to j} (\mu_G(i, j) - \mu_H(i, j)) \right|
    &\le  \left|\left(\sum_{(i\to j)\in \partial (S - \{u\})} \COEF'_{i\to j} (\mu_G(i, j) - \mu_H(i, j))\right) + \COEF_u\cdot 0 \right| \tag{vertex $u$ is saturated} \\
    &\le |c_V^G(i) - c_V^H(i)|. \tag{by induction hypothesis}
\end{align*}

By choosing $S=V\setminus T$ and coefficients $\COEF_{i\to j} = \mathrm{sgn}(\mu_G(i, j) - \mu_H(i, j))$ for every edge $(i\to j)\in \partial S$, 
we conclude that 
\begin{align*}
 &\ \ \ \ \sum_{i} \left| \left(c_V^G(i) -  \sum_j \mu_G(i, j)\right) - \left(c_V^H(i) - \sum_j \mu_H(i, j)\right)\right|\\
 &= \sum_{i\neq s}\left|\sum_j\mu_G(i, j) - \sum_j\mu_H(i, j)\right| \tag{for all $i\neq s$, $c_V^G(i)=c_V^H(i)$, and $s$ is saturated}\\
    &=\sum_{i\in T} \left|\sum_j\mu_G(i, j) - \sum_j\mu_H(i, j)\right|\\
    &=\sum_{(i\to j)\in \partial S} \COEF_{i\to j}(\mu_G(i, j) - \mu_H(i, j)) \tag{use $\COEF_{i\to j}$ to remove the absolute operation}\\
    &=\left|\sum_{(i\to j)\in \partial S} \COEF_{i\to j}(\mu_G(i, j) - \mu_H(i, j))\right| \tag{this sum is positive}\\
    &\le |c_V^G(s) - c_V^H(s)|.
\end{align*}

\end{proof}

Similarly, by a reduction to vertices changes, we have the bound for the edge changes.

\begin{lemma}\label{lem:rising-tide-diff-by-one-edge}
Assume that $G$ and $H$ differ only in the capacity 
of one edge $(s, t)\in E$.
Assume that $D:=D_G=D_H$.
Then,
\[
\sum_{i} \left| \left(c_V^G(i) -  \sum_j \mu_G(i, j)\right) - \left(c_V^H(i) - \sum_j \mu_H(i, j)\right)\right|\le 2|c_E^G(s, t) - c_E^H(s, t)|.
\]
\end{lemma}

\begin{proof}
The proof is by reduction to \autoref{lem:rising-tide-diff-by-one-vertex}.
Create $G'$ by subdividing $(s,t)$ into $(s,x),(x,t)$
with $c_E^{G'}(s,x)=c_E^{G'}(x,t)=\infty$
and $c_V^{G'}(x)= 2c_E^G(s,t)$.
Create $H'$ from $H$ in the same way.
Since $D_G=D_H$, the same vertices must be saturated
in both, and in particular, among $s,t,$ and $(s,t)$,
both executions saturate the same element first.
If they both saturate $s$ or $t$ first, then the capacity
of $(s,t)$ has no influence on the execution 
and $\mu_G=\mu_H$.  If they both saturate $(s,t)$ first,
then the executions on $\{G,H\}$ 
proceed identically to 
their counterpart executions on $\{G',H'\}$.
Note that $G',H'$ differ in 
one vertex capacity, with
$|c_V^{G'}(x)-c_V^{H'}(x)|=2|c_E^{G}(s,t)-c_E^{H}(s,t)|$.
The lemma then follows from 
\autoref{lem:rising-tide-diff-by-one-vertex}
applied to $G',H'$.
\end{proof}

\ignore{
\begin{proof}[Proof Sketch.]
The proof is extremely similar to the proof in \autoref{lem:rising-tide-diff-by-one-vertex}, so we only describe the necessary rewriting here.
First, we rewrite \autoref{observation:path-to-changing-vertex} so that for any edge $(i, j)$ that $\mu_G(i, j)\neq \mu_H(i, j)$, there exists a (possibly empty) path in the dependency graph $D$ from at least one endpoint in $\{i, j\}$ to one of the vertices in $\{s, t\}$. This can be proved again by induction (proof omitted because it is the same as the proof in \autoref{observation:path-to-changing-vertex}.)
Furthermore, let $T$ be the set of unsaturated vertices in the subgraph $\DDIFF$ of $D$ where the edges do change their fractional matching values.
We have the similar claim that for any downward-closed set $S$ such that $\{s, t\} \subseteq S\subseteq V\setminus T$, and any set of coefficients 
$\{\COEF_{i\to j} \in [-1, 1]\}_{(i\to j)\in \partial S}$
we have
\[
\left|\sum_{(i\to j)\in \partial S} \COEF_{i\to j} (\mu_G(i, j) - \mu_H(i, j)) \right| \le 2|c_E^G(s, t) - c_E^H(s, t)|.
\]

The base case is where $S=\{s, t\}$ (this is where a leading $2$ occurs if both $s$ and $t$ are saturated) and the inductive part of the proof is exactly the same idea as in \autoref{lem:rising-tide-diff-by-one-vertex}.
\end{proof}
}

We can now prove \autoref{lem:rising-tide-output}.

\begin{proof}[Proof of \autoref{lem:rising-tide-output}]
Imagine continuously transforming $(c_V^G,c_E^G)$ into $(c_V^H,c_E^H)$ 
by modifying one vertex capacity or one edge capacity at a time.  
In this continuous process there are two types of 
\emph{breakpoints} to pay attention to.  
The first is when we switch from transforming one capacity to another,
and the second is when the dependency graph changes. 
Let $G=G_0,G_1,\ldots,G_k=H$ be the sequence of graphs at these breakpoints.
Up to a tie-breaking perturbation, we can assume each pair $(G_i,G_{i+1})$
differ in one edge or vertex capacity, and have the same dependency graph.
By~\autoref{lem:continuity-lemma} the objective function is continuous
in the input, and does not have any discontinuities at breakpoints.
By \autoref{lem:rising-tide-diff-by-one-vertex} and \autoref{lem:rising-tide-diff-by-one-edge}
the objective function is bounded by 
$\sum_i (\eta_V(i)+2\eta_E(i)) = \eta_V + 2\eta_E$.
\end{proof}

\end{document}